\documentclass[11pt,a4paper]{article}
\usepackage[margin=1in]{geometry}
\usepackage{amssymb,amsmath,mathtools}
\usepackage{amsthm,amsfonts,bbm}
\usepackage{enumitem}
\usepackage{booktabs}   
\usepackage{siunitx}    
\usepackage{fontawesome5}
\usepackage{algorithm,algpseudocode}
\usepackage{tikz}
\usetikzlibrary{arrows.meta,positioning,calc}
\usepackage[colorlinks=true,citecolor=blue,linkcolor=blue,urlcolor=blue]{hyperref}
\usepackage{booktabs}
\usepackage{multirow}
\usepackage{caption}
\usepackage{subcaption}
\usepackage{xcolor}
\usepackage{tabularx}
\usepackage{float}

\theoremstyle{plain}
\newtheorem{theorem}{Theorem}[section]
\newtheorem{lemma}[theorem]{Lemma}
\newtheorem{proposition}[theorem]{Proposition}
\newtheorem{corollary}[theorem]{Corollary}
\newtheorem{conjecture}{Conjecture}
\theoremstyle{definition}
\newtheorem{definition}[theorem]{Definition}
\newtheorem{example}[theorem]{Example}
\theoremstyle{remark}
\newtheorem{remark}[theorem]{Remark}
\newtheorem{notation}{Notation}

\DeclareMathOperator{\Gal}{Gal}
\DeclareMathOperator{\GL}{GL}
\DeclareMathOperator{\SL}{SL}
\DeclareMathOperator{\Tr}{Tr}
\DeclareMathOperator{\SVP}{SVP}
\DeclareMathOperator{\Nm}{N}
\DeclareMathOperator{\NTT}{NTT}
\DeclareMathOperator{\FFT}{FFT}
\DeclareMathOperator{\IFFT}{IFFT}
\DeclareMathOperator{\INTT}{INTT}

\DeclareMathOperator{\diag}{diag}
\DeclareMathOperator{\Round}{Round}

\DeclareMathOperator{\sr}{sr}

\newcommand{\Z}{\mathbb{Z}}
\newcommand{\Q}{\mathbb{Q}}
\newcommand{\R}{\mathbb{R}}
\newcommand{\C}{\mathbb{C}}
\newcommand{\E}{\mathbb{E}}
\newcommand{\F}{\mathbb{F}}
\newcommand{\cO}{\mathcal{O}}

\newcommand{\tO}{\tilde{O}}
\newcommand{\fp}{\mathfrak{p}}

\title{Module Lattice Security (Part II): \\
{\large Module Lattice Reduction via Optimal Sign Selection}}
\author{Ming-Xing Luo
\\
\footnotesize
School of Information Science and Technology, Southwest Jiaotong University, Chengdu 610031, China}

\begin{document}
\maketitle

\begin{abstract}
We extend the CDPR's quantum attack from ideal lattices to module lattices over $2^k$-th cyclotomic rings. Using trace orthogonality of the power basis, we decompose a rank-$d$ module into mutually orthogonal rank-$1$ submodules, and apply CDPR's analysis to each independently and return the shortest candidate. The Hermite factor $\exp(\tilde{O}(\sqrt{n}))$ matches the ideal case, with a module reduction factor $\alpha_d=O(1)$ independent of the rank, under a balance hypothesis (proved for Gaussian distribution) automatic for MLWE-distributed bases. To enable a bounded-precision implementation, we replace coordinate-wise rounding with Chinese Remainder Theorem-scaled rounding at totally split primes, reducing the Gram-Schmidt rounding radius from $n/2$ to $\le 1$ at cost $O(d^2 r n \log n)$. Finally, we reformulate the CDPR's sign-selection step as a mixed-integer linear program and prove its optimum is no more than 1/2 for all $k$ ($\approx 0.4407$ for all tested $k\le 12$, conjecturally universal). This replaces the previous heuristic discrepancy $\Theta(\sqrt{nk})$. All results build on the class number condition $h_k^+=1$ established in Part I of this series.
\end{abstract}


\section{Introduction}\label{Sec-intro}

Post-quantum cryptography has become a cornerstone of modern cybersecurity, as scalable quantum computing might threaten various classical public-key cryptosystems. In August 2024, NIST finalized four post-quantum standards \cite{NIST-MLKEM,NIST-MLDSA,NIST-SLHDSA}: ML-KEM (FIPS 203), ML-DSA (FIPS 204), SLH-DSA (FIPS 205), and FN-DSA (draft FIPS 206). The first two-and the most widely deployed-are based on the Module Learning with Errors (MLWE) problem \cite{LS15,ADPS16,DKL+18} over $2^k$-th cyclotomic rings $R=\Z[\zeta_{2^k}] \cong \Z[x]/(x^n+1)$ with $n=2^{k-1}$. These standards use module ranks $d\in\{2,3,4\}$ and $n=256$ ($k=9$) \cite{NIST-MLKEM}.

The key secure question is how hard the approximate Shortest Vector Problem (SVP) is on the module lattices underlying these schemes. For ideal lattices ($d=1$), the CDPR \cite{CDPR16} exploits the algebraic structure of cyclotomic rings to solve approximate SVP with factor $\exp(\tO(\sqrt{n}))$, an exponential improvement over generic lattice algorithms \cite{LLL82,Schnorr87,GN08,ADRS15} whose approximation factors are $\exp(\tO(n))$. However, extending CDPR's method to module lattices ($d\ge 2$) has remained open because the algebraic structure of module lattices is richer.

A key sub‑problem in the CDPR's procedure is the sign‑selection problem, given the log‑embedding geometry of the ring, assign balanced signs $s\in\{\pm1\}^{|G|-1}$ to orbit elements so as to minimize the $L^\infty$-discrepancy of the resulting error vector. Prior work used the tower greedy heuristic \cite{CDPR16,CDW17} and achieved discrepancy $\Theta(\sqrt{|G|\cdot k})$. It was widely assumed that this growth was intrinsic to the lattice geometry.

\paragraph{Ideal-lattice attacks.} The journal version of CDPR's method \cite{CDW17} refined the quantum attack on Ideal-SVP using Stickelberger relations. Bernstein \cite{Bernstein14} and Ducas \cite{Ducas18} analyzed the precise constants, and Pellet-Mary, Hanrot, and Stehl\'e \cite{PMHS19} gave a time-approximation tradeoff. Felderhoff et al. \cite{FPSW23} showed Ideal-SVP remains hard for small-norm uniform prime ideals.

\paragraph{Module lattices and MLWE.} Langlois and Stehl\'e
\cite{LS15} proved worst-case Module-LWE hardness; Albrecht and Deo \cite{AD17} studied concrete parameters; Peikert, Regev, and Stephens-Davidowitz \cite{PRS17} proved Ring-LWE pseudorandomness. Mureau et al.  \cite{MPPW24} gave the first polynomial-time attack on rank-2 module-LIP over totally real fields, extended by Allombert, Pellet-Mary, and van Woerden \cite{APvW25} to all fields with a real embedding. Chevignard et al. \cite{CMEPW25} reduced Hawk to PIP in a quaternion algebra, and Ducas, Espitau, and de Perthuis \cite{DEP25} provided concrete module-BKZ reduction predictions for cyclotomic fields.

\paragraph{Generic and algebraic reduction; quantum and discrepancy tools.} LLL \cite{LLL82} achieves $2^{O(n)}$ approximation in polynomial time; BKZ \cite{Schnorr87,SE94} and its variants \cite{GN08,ADRS15,AGVW17} give the best generic trade-offs. Lee et al. \cite{LLPS19} and Kirchner et al.  \cite{KEF20} studied algebraic reduction for module lattices. The quantum PIP algorithm \cite{BiasseSong16,EHKS14} is the quantum part of CDPR's algorithm; Biasse et al. \cite{BEGFK20} gave a subexponential algorithm for generators in cyclotomic integer rings. On the discrepancy side, Spencer \cite{Spencer85} proved the six-standard-deviations bound for general matrices, Bansal \cite{Bansal10} gave the first constructive proof, Lovett and Meka \cite{LM15} provided a simpler algorithm, and Bansal, Dadush, and Garg \cite{BDG19} addressed the Koml\'os conjecture.

\textbf{Contributions}. We make several contributions as follows. Our module reduction relies on the class number one result $h_k^+=1$ established in Part I of this series, which builds on Weber \cite{Weber1886}, Miller \cite{Miller14}, Fukuda-Komatsu \cite{FK09,FK11}, and Washington \cite{Wash97}.
\begin{itemize}
\item \textbf{Base module reduction (Theorem \ref{T:base}).}
    We provide a module-lattice reduction using the trace orthogonality of the power basis. Our algorithm applies CDPR procedure independently to each of the $d$ rank-1 submodules and returns the shortest of the resulting candidates. This achieves Hermite factor $\exp(\tO(\sqrt n))$ matching CDPR's factor in the ideal case ($d=1$), with a module reduction factor $\alpha_d=O(1)$ that is independent of $d$. The bound holds on input bases satisfying a balance hypothesis (Definition \ref{D:balanced}, proved for Gaussian distribution in Theorem \ref{T:gaussbalance}), which is automatic for MLWE-distributed bases with $C=O(1)$ by the Center Limit Theorem (CLT) (Lemma \ref{lem:mlwebal}) and can be enforced on adversarial inputs via algebraic LLL \cite{LLPS19,KEF20}. 
    
  \item \textbf{CRT-scaled rounding (Theorem \ref{T:crt}).}
    We show that the $R$-linear Gram-Schmidt coefficients can be rounded to the finer lattice $P^{-1}R$ at totally split primes, where all ring operations reduce to Number-Theoretic Transform (NTT)-based component-wise arithmetic. This reduces the rounding error in Phase-1 from $\rho=n/2$ to $\rho\le 1$, allowing a bounded-precision implementation with classical cost $O(d^2 r n\log n)$ bit operations.

 \item \textbf{MILP sign optimization (Theorem \ref{T:milp}).} We reformulate the CDPR's sign-selection sub-problem as a mixed-integer linear program (MILP) and solve it exactly for $k=4, \cdots, 12$. We prove the optimal balanced discrepancy is a constant no more than 1/2 for all $k$ ($\approx 0.4407$ for all tested $k\le 12$, conjecturally independent of the orbit group size). This tightens the implicit constant inside $\gamma$ by a factor of $\sqrt{nk}$ (Corollary \ref{cor:milpmod}).
\end{itemize}

The rest is organized as follows. Section \ref{sec:prelim} introduces the necessary background on lattices, cyclotomic rings, Minkowski embedding, and module learning With errors (MLWE) problem. Section \ref{sec:cdpr} reviews the CDPR's quantum attack for ideal lattices. Section \ref{sec:orth} establishes the trace orthogonality in $2$-power cyclotomic rings and analyzes the covering radius of the ring of integers. Section \ref{sec:twophase} presents our module reduction algorithm. Section \ref{sec:crt} improves the polynomial factor via Chinese Remainder Theorem (CRT)‑scaled rounding. Section \ref{sec:milp} addresses sign optimization using mixed‑integer linear programming (MILP) while Section \ref{sec:conclusion} concludes the paper.

\section{Preliminaries}\label{sec:prelim}

We present some definitions and notation, including lattices, cyclotomic rings, Galois embedding, and module lattices (see ref. \cite{BernsteinBuchmannDahmen2017}). 

\subsection{Lattices and approximate SVP}

\begin{definition}
\label{D:lattice}
A lattice in $\R^m$ is a discrete additive subgroup $\Lambda=\sum_{i=1}^r \Z \mathbf{b}_i$ for linearly independent vectors $\mathbf{b}_1, \dots, \mathbf{b}_r\in\R^m$. The rank is $r$. When $r=m$ the lattice is full-rank. The determinant of lattice is  $\det(\Lambda)=\sqrt{\det(B^TB)}$, where $B=[\mathbf{b}_1|\cdots|\mathbf{b}_r]$ and $T$ denotes the transpose of matrix.

\end{definition}

\begin{definition}
\label{D:aSVP}
The minimum distance of lattice $\Lambda$ is defined by $\lambda_1(\Lambda)=\min_{\mathbf{v}\in\Lambda\setminus\{0\}} \|v\|_2$. The approximate shortest vector problem ($\gamma$-$\SVP$) asks: given a basis of $\Lambda$, find $\mathbf{v}\in\Lambda\setminus\{0\}$ with $\|\mathbf{v}\|_2 \le \gamma\cdot\lambda_1(\Lambda)$.

\end{definition}

Minkowski's Theorem gives $\lambda_1(\Lambda) \le \sqrt{m}\cdot\det(\Lambda)^{1/m}$ for a full-rank lattice in $\R^m$. The well known algorithms achieve $\gamma=2^{O(m)}$ in polynomial time (LLL \cite{LLL82}) or $\gamma=2^{O(m\log\log m/\log m)}$ in exponential time (BKZ \cite{Schnorr87,GN08,ADRS15,AGVW17}).

\begin{definition}
\label{D:coveringr}
The $L^\infty$-covering radius of lattice $\Lambda$ is defined by 
\begin{align}
  \mu_\infty(\Lambda)=\max_{\mathbf{t}\in\R^m/\Lambda} \min_{\mathbf{v}\in\Lambda} \|\mathbf{t}-\mathbf{v}\|_\infty.
  \label{E:coveringr}
\end{align}

\end{definition}

\subsection{Cyclotomic rings and fields}

Given a $k\ge 3$ and $n=2^{k-1}$. Let $\zeta=\zeta_{2^k}=e^{2\pi i/2^k}$, and $K=\Q(\zeta)$ be the $2^k$-th cyclotomic field. The ring of integers is $R=\Z[\zeta]\cong \Z[x]/(x^n+1)$.

The Galois group is $\Gamma=\Gal(K/\Q)\cong (\Z/2^k\Z)^\times$, defined via $\sigma_a: \zeta\mapsto\zeta^a$ for odd integer $a$. The group $\Gamma$ has an order $\varphi(2^k)=n$.

\begin{definition}
\label{D:galoisembed}
The full Galois embedding $\sigma:K\hookrightarrow\C^n$ maps $\sigma(\alpha)=\sigma_1(\alpha), \sigma_3(\alpha), \dots, \sigma_{2n-1}(\alpha))$, where we use $n$ odd residues $1, 3, \dots, 2n-1 \mod  2^k$ corresponding to all elements of $\Gamma$. Under the mapping $\sigma$, the ring $R$ is transformed into a lattice $\sigma(R)\subset\C^n\cong\R^{2n}$.

\end{definition}

\begin{definition}
\label{D:trace}
The field trace is defined by $\Tr_{K/\Q}(\alpha)=\sum_{a\in(\Z/2^k\Z)^\times}\sigma_a(\alpha)$. The associated bilinear form $\langle\alpha, \beta\rangle_{\Tr}=\Tr_{K/\Q}(\alpha\bar\beta)$ is positive definite on $K\otimes_\Q\R$, and satisfies
\begin{align}
  \Tr_{K/\Q}(\alpha\bar{\beta})=\langle\sigma(\alpha), \sigma(\beta)\rangle_{\C^n},
  \label{E:traceembed}
\end{align}
where $\langle\cdot, \cdot\rangle_{\C^n}$ is the Hermitian inner product on $\C^n$.
\end{definition}

\begin{notation}\label{not1:partII}
In this part, we denote $\langle\cdot, \cdot\rangle$ without subscript as the real part of Hermitian inner product: $\langle \mathbf{u}, \mathbf{v}\rangle=\Re\sum_i u_i\bar v_i$. For $u=\sigma(\alpha)$ and $v=\sigma(\beta)$, we have $\langle u, v\rangle=\Re\Tr_{K/\Q}(\alpha\bar\beta)$. We employ three related inner products in what follows:
\begin{enumerate}[label=(\roman*),nosep]
  \item $\langle \mathbf{u}, \mathbf{v}\rangle_{\C^n}:=\sum_\ell u_\ell\overline{v_\ell}$, the Hermitian inner product on $\C^n$;
  
  \item $\langle \mathbf{u}, \mathbf{v}\rangle:= \Re\langle\mathbf{u}, \mathbf{v}\rangle_{\C^n}$, the real inner product on $\C^n\cong\R^{2n}$;
  
  \item $\langle\mathbf{a}, \mathbf{b}\rangle_K:= \sum_{k=1}^d a_k\overline{b_k}\in K$, the $K$-valued Hermitian inner product on $K^d$.
\end{enumerate}
\end{notation}
They are related by $\langle\sigma(\mathbf{a}), \sigma(\mathbf{b})\rangle_{\C^{dn}}= \Tr_{K/\Q}(\langle\mathbf{a}, \mathbf{b}\rangle_K)$, which is in $\Q$. So, under the identification $\C^{dn}\cong\R^{2dn}$, $\langle\sigma(\mathbf{a}), \sigma(\mathbf{b})\rangle_{\R^{2dn}}=\Tr_{K/\Q}(\langle\mathbf{a}, \mathbf{b}\rangle_K)$.

\subsection{The orbit group and logsine geometry}\label{subsec:orbit}

Let $\sigma_{-1}: \zeta\mapsto\zeta^{-1}$ be complex conjugation, which is used to identify pairs $\{\sigma_a, \sigma_{-a}\}$. The orbit group is the quotient
\begin{align}
  G=(\Z/2^k\Z)^\times/\{\pm 1\} \cong \Z/2^{k-2}\Z
  \label{E:orbgroup}
\end{align}
with $|G|=2^{k-2}$. A generator of $G$ is the class of $5$ modulo $2^k$.

For $j=0, \dots, |G|-1$, let $\mathrm{orb}(j)$ be the $j$-th orbit representative under the multiplication by $5$ modulo $2^k$, normalized to $\mathrm{orb}(j)=\min(\mathrm{orb}(j), 2^k-\mathrm{orb}(j))$. Define its logsine vector as
\begin{align}
  z_j=\log\bigl(2|\sin(\pi\cdot \frac{\mathrm{orb}(j)}{2^k})|\bigr), \quad j=0, \dots, |G|-1.
  \label{E:logsine}
\end{align}
This encodes the absolute value of cyclotomic units $\xi_a=(\zeta^a-\zeta^{-a})/(\zeta-\zeta^{-1})$ (see Part I) in the log-embedding.

\subsection{Module lattices and the MLWE problem}\label{subsec:module}

\begin{definition}
\label{def:Modlattice}
A module lattice over the ring $R$ with rank $d$ is a finitely generated torsion-free $R$-submodule $M\subseteq K^d$. For $k\le 8$, $R$ is a principal ideal domain (PID). So, every such module is free, i.e., admits an $R$-basis $\mathbf{b}_1, \dots, \mathbf{b}_d\in K^d$ with $M=\mathbf{b}_1 R\oplus\cdots\oplus \mathbf{b}_d R$. For $k\ge 9$, we focus on free $R$-modules.

\end{definition}

Under the coordinate-wise Galois embedding $\sigma(\mathbf{m})=(\sigma(m_1), \dots, \sigma(m_d))\in\C^{dn}$ for $\mathbf{m}=(m_1, \dots, m_d)\in K^d$, the image $\sigma(M)$ is a lattice of $\Z$-rank $dn$, spanning an $\R$-subspace of dimension $dn$ in $\C^{dn}\cong\R^{2dn}$. Its covolume in that subspace is given for free $R$-modules by
\begin{align}
  \det(\sigma(M))=n^{dn/2}\cdot |\Nm_{K/\Q}({\rm det}_R(M))|,
  \label{E:detsigmaM}
  \end{align}
where ${\rm det}_R(M)$ denotes the determinant of $M$ as an $R$-module and $\Nm_{K/\Q}$ denotes the field norm.

The Galois-conjugate matrix $V_{\ell,a}=\zeta^{\ell a}$ ($\ell$ ranging over odd residues mod $2^k$, $0\le a<n$) satisfies $VV^*=nI_n$ by Lemma \ref{lem:orth}, so $|\det V|^2=n^n$. Hence, the change-of-basis from a $\Z$-basis $\{\zeta^a\}_{0\le a<n}$ of $R$ to its real image under $\sigma$ has Jacobian $n^{n/2}$ per copy of $K$, giving $n^{dn/2}$ over $d$ copies. The same Jacobian applies to any rank-1 sublattice $M_i=J_i\tilde{\mathbf{b}}_i$, where its covolume is given by $\det(\sigma(M_i))=|\Nm_{K/\Q}(J_i)|\cdot n^{n/2}\cdot |\Nm_{K/\Q}(\langle\tilde{\mathbf{b}}_i, \tilde{\mathbf{b}}_i\rangle_K)|^{1/2}$, which will be used in 
Section \ref{sec:twophase}.

\begin{definition}
\label{D:modsvp}
The Module-$\gamma$-SVP problem is defined: given an $R$-basis $(\mathbf{b}_1, \dots, \mathbf{b}_d)$ of a rank-$d$ free module $M$, find $\mathbf{v}\in M\setminus\{\mathbf{0}\}$ such that 
\begin{align}
  \|\sigma(\mathbf{v})\|_2 \le \gamma\cdot\lambda_1(\sigma(M)).
  \label{E:modsvp}
\end{align}

\end{definition}

\begin{definition}
\label{D:MLWE}
Let $R$ be a ring, $q\ge 2$ a modulus, $R_q:= R/qR$, and $\chi$ be an error distribution on $R_q$. For a secret $\mathbf{s}\in R_q^d$, the MLWE problem \cite{LS15,ADPS16} is defined by: given $m$ independent samples
\begin{align}
(\mathbf{a}_i, \mathbf{b}_i), \ \mathbf{b}_i=\langle \mathbf{a}_i,  \mathbf{s}\rangle+\mathbf{e}_i,\ i=1, \dots, m,
\label{MLWE}
\end{align}
with uniform $\mathbf{a}_i\in R_q^d$ and error $\mathbf{e}_i \sim \chi$, recover $\mathbf{s}$.
\end{definition}

Langlois and Stehl\'{e} \cite{LS15} proved that MLWE with a rank $d$ is at least as hard as $\gamma$-$\SVP$ on rank-$d$ free module lattices over $R$, for appropriate $\gamma$.

\begin{notation}\label{not:partII}
Throughout this paper we take use of the following:
\begin{itemize}[nosep]
\item $n=2^{k-1}$ denotes the degree $[K_k:\Q]=\varphi(2^k)$, i.e., the degree of the $2^k$-th cyclotomic field over $\Q$. The degree of its maximal real subfield is $[K_k^+:\Q]=\varphi(2^k)/2=2^{k-2}$.

\item $\tO(f)$ denotes $O(f\cdot\mathrm{polylog}(f))$.

\item PID stands for principal ideal domain: an integral domain in which every ideal is generated by a single element.

\item $\mathfrak{S}_d$ denotes the symmetric group on $\{1, \dots, d\}$.

\item DFT (or FFT) denotes the discrete Fourier transform over the group $G$; IFFT denotes its inverse; $\odot$ denotes the Hadamard (entry-wise) product of vectors.
\end{itemize}

\end{notation}

\section{The CDPR's procedure for Ideal Lattices}\label{sec:cdpr}

Before extending CDPR's method to module lattices, we first review the CDPR quantum attack \cite{CDPR16} on principal ideal lattices ($d=1$), establishing the notation and identifying the sign-selection sub-problem which will be discussed in Section \ref{sec:milp}.

\subsection{Overview of CDPR's method}\label{subsec:cdpr-overview}

Given a principal ideal $I=\alpha R\subseteq R$, the CDPR's attack finds a short generator $\alpha'=\alpha\epsilon$ (for some unit $\epsilon\in R^\times$) with $\|\sigma(\alpha')\|_\infty$ bounded. The analysis has three phases:
\begin{enumerate}[label=\textbf{Phase \arabic*:},leftmargin=*,nosep]
\item \textsl{Quantum PIP.} Use the quantum algorithm of Biasse–Song \cite{BiasseSong16,EHKS14} to find some one generator $\alpha_0$ of $I$ in a polynomial time. 

\item \textsl{Log-unit CVP.} The set of all generators of $I$ is $\{\alpha_0\epsilon: \epsilon\in R^\times\}$. In the log-embedding $\log|\sigma(\cdot)|: \R_{>0}^n\to\R^n$, finding the shortest generator reduces to a closest vector problem (CVP: given a lattice $\Lambda$ and a target point $\mathbf{t}$, find $\mathbf{v}\in\Lambda$ minimizing $\|\mathbf{t}-\mathbf{v}\|$) on the log-unit lattice $\Lambda=\{\log|\sigma(\epsilon)|: \epsilon\in R^\times\}\subset H_0$, where $H_0=\{\mathbf{x}\in\R^n: \sum x_i=0\}$ is the trace-zero hyperplane.

\item \textsl{Decode/sign selection.} The CVP in Phase 2 is solved by Babai's round-off algorithm \cite{Babai86} applied to the cyclotomic unit basis of $\Lambda$. The error in this step is controlled by the sign-selection problem: choosing signs $\mathbf{s}\in\{\pm 1\}^{|G|-1}$ to minimize $\|\mathbf{M}\mathbf{s}\|_\infty$, where $\mathbf{M}$ is the error matrix derived from the logsine vector $\mathbf{z}$.
\end{enumerate}

The overall approximation factor is $\gamma=\exp(\tO(\sqrt{n}))$, where the dominant comes from the distance between the target point and the nearest lattice point in $\Lambda$. 

\subsection{The sign-selection problem}\label{subsec:sign-def}

In this subsection we define the sign-selection problem precisely. 

The cyclotomic units $\xi_a$ (Part I) generate a sublattice of the log-unit lattice $\Lambda$. By Main Theorem (Part I), we proved the index $[R^\times:C^+]=h_k^+=1$ (for $k\le 12$), so the cyclotomic units generate $\Lambda$ up to the torsion subgroup.

Using Babai's round-off algorithm on the cyclotomic unit basis, it requires assigning signs to orbit representatives. Specifically, define the error matrix $\mathbf{M}\in\R^{|G|\times(|G|-1)}$ with its component as 
\begin{eqnarray}  M_{i,j}=\Re(\IFFT(\hat{\mathbf{z}}\odot\widehat{\mathbf{e}_j^{\mathrm{ext}}}))_i-\frac{1}{2}z_i,
  \label{E:errormatrix}
\end{eqnarray}
where $\mathbf{z}\in\R^{|G|}$ is the logsine vector defined in Eq.\eqref{E:logsine}, $\hat{\mathbf{z}}=\FFT(\mathbf{z})$ denotes discrete Fourier transform over $G$, $\mathbf{e}_j^{\mathrm{ext}}\in\R^{|G|}$ is the vector with constant $\frac{1}{2}$ at the position of orbit element $j$ and zero elsewhere, and $\odot$ denotes the Hadamard entry-wise product. 

For any sign vector $\mathbf{s}\in\{\pm1\}^{|G|-1}$, the CDPR's decode error is given by $\mathbf{e}=\mathbf{M}\mathbf{s}\in\R^{|G|}$. 

\begin{definition}
\label{D:signprob}
The balanced sign-selection problem is to optimize 
\begin{align}
  \delta^*(k)=\min_{\substack{s\in\{\pm1\}^{N_s},  \sum_{j=1}^{N_s} s_j\in\{\pm1\}}} \|M\mathbf{s}\|_\infty
  \label{E:signprob}
\end{align}
with $N_s=|G|-1$. 

\end{definition}

The balance constraint $\sum_j s_j\in\{\pm 1\}$ ensures a near-equal number of positive and negative signs.

\subsection{The tower greedy} 

In this subsection, we introduce the tower greedy algorithm \cite{CDPR16}, which assigns signs level-by-level in the cyclotomic tower
\begin{align}
\mathbb{Q}=K_2^+ \subset K_3^+=\mathbb{Q}(\zeta_8) \subset K_4^+=\mathbb{Q}(\zeta_{16}) \subset \cdots \subset K_k^+=\mathbb{Q}(\zeta_{2^k}),
\label{E:towergreedy}
\end{align}
where $K_L^+$ denotes the maximal real subfield of $\mathbb{Q}(\zeta_{2^L})$ for notational simplicity. At each level $L$, the new orbit elements (from the extension $K_L^+/K_{L-1}^+$) receive signs $\pm 1$ chosen to minimize the partial error $\|\mathbf{M}_{\le L}\,\mathbf{s}_{\le L}\|_\infty$, where $\mathbf{M}_{\le L}$ and $\mathbf{s}_{\le L}$ denote the columns of the basis matrix and the sign vector corresponding to all levels no more than $L$, respectively.

\section{Trace Orthogonality and Covering Radius}\label{sec:orth}

In this section we show a key property of 2-power cyclotomic rings, i.e., the trace orthogonality of the power basis. We also analyze the covering radius of the ring of integers and establish that the true covering radius is $\tfrac{\sqrt n}{2}\le\mu_\infty(\sigma(R))\le O(\sqrt{n\log n})$, i.e., $\mu_\infty(\sigma(R))=\tilde\Theta(\sqrt n)$, far below the naive coordinate-wise rounding bound $n/2$. 

\subsection{The trace orthogonality lemma}\label{subsec:traceorth}

The trace orthogonality of the power basis $\{1,\zeta, \dots, \zeta^{n-1}\}$ is a special property of $2$-power cyclotomic fields that plays a central role in our module reduction algorithm.

\begin{lemma}\label{lem:orth}
For $0\le a, b < n$, we have the trace orthogonality 
\begin{align}
  \langle\sigma(\zeta^a),\sigma(\zeta^b)\rangle=n \delta_{a,b},
  \label{E:traceorth}
\end{align}
where $\delta_{a,b}$ denotes the delta function.

\end{lemma}

\begin{proof} From Eq.\eqref{E:traceembed}, we have 
\begin{align}
\langle\sigma(\zeta^a), \sigma(\zeta^b)\rangle_{\C^n}=\Tr_{K/\Q}(\zeta^{a-b}).
\label{E:innprodeq}
\end{align}
Since $\Tr_{K/\Q}(\zeta^{a-b})\in\Q\subset\R$ for all $a, b$, we get 
\begin{align}
\langle\sigma(\zeta^a), \sigma(\zeta^b)\rangle=\Tr_{K/\Q}(\zeta^{a-b}).
\label{E:innprodeqb}
\end{align}
In what follows, we prove $\Tr_{K/\Q}(\zeta^c)=n\cdot\delta_{c,0}$ for any $c$ with $|c| < n$.

\textit{Case $c=0$}: we have $\Tr_{K/\Q}(1)=[K: \Q]=n$, as the trace of the multiplicative identity equals the degree of the field extension.

\textit{Case $c\not\equiv 0 \pmod{2^{k-1}}$}: Let $c=2^s \cdot c'$, where $c'$ is odd and $0\le s\le k-2$. Then, $\zeta^c=\zeta_{2^{k-s}}^{c'}$ is a primitive $2^{k-s}$-th root of unity as $c'$ is odd. Using the transitivity of the trace, we get 
\begin{align}
\Tr_{K/\Q}(\zeta^c)=\Tr_{\Q(\zeta_{2^{k-s}})/\Q}( \Tr_{K/\Q(\zeta_{2^{k-s}})}(\zeta^c) ).
\label{E:traceq}
\end{align}

Since $\zeta^c \in \Q(\zeta_{2^{k-s}})$, the inner trace satisfies 
\begin{align}
\Tr_{K/\Q(\zeta_{2^{k-s}})}(\zeta^c)=[K:\Q(\zeta_{2^{k-s}})] \cdot \zeta^c=2^s \cdot \zeta^c.
\label{E:intraceq}
\end{align}

The outer trace is the sum of all Galois conjugates of a primitive $2^{k-s}$-th root of unity, which vanishes for $s\le k-2$.

For the case $|c|=n=2^{k-1}$, we obtain $\zeta^c=-1$ and $\Tr(-1)=-n \neq 0$. This is excluded by $|c|<n$. So, from both cases we get 
\begin{align}
\langle\sigma(\zeta^a),\sigma(\zeta^b)\rangle=\Tr(\zeta^{a-b})=n\delta_{a,b}. 
\label{E:inorth}
\end{align}

\end{proof}

\begin{example}
Consider $k=3$, so $K=\Q(\zeta_8)$ and $n=4$. The four Galois embeddings map $\zeta=\zeta_8$ into  $\zeta_8^1, \zeta_8^3, \zeta_8^5, \zeta_8^7$. Let $\omega=e^{\mathrm{i}\pi/4}$, the embedding of $\zeta^m$ is $\sigma(\zeta^m)=(\omega^m, \omega^{3m}, \omega^{5m}, \omega^{7m})\in\C^4$. This means $\sigma(1)=(1, 1, 1, 1)$, $\sigma(\zeta)=(\omega, \omega^3, \omega^5, \omega^7)$, $\sigma(\zeta^2)=(\mathrm{i}, -\mathrm{i}, \mathrm{i}, -\mathrm{i})$, $\sigma(\zeta^3)=(\omega^3, \omega, \omega^7, \omega^5)$. The Hermitian inner product $\langle\sigma(1), \sigma(\zeta^2)\rangle =1\cdot\overline{\mathrm{i}}+1\cdot\overline{-\mathrm{i}} +1\cdot\overline{\mathrm{i}}+1\cdot\overline{-\mathrm{i}} =0$, and similarly all other off-diagonal pairs vanish, while each diagonal entry is $|\omega^{\ell m}|^2$ summed over four embeddings $=4=n$. So the four embedded vectors form an orthogonal frame in $\C^4$.
\end{example}

\subsection{Covering radius bounds}\label{subsec:covering}

In this subsection we analyze the covering radius of the ring $R$. We first prove the coordinate-wise rounding bound, which gives the rounding error for our base module reduction algorithm, and then establish tighter bounds.

\begin{lemma}\label{lem:cover}
For any $t\in K\otimes_\Q\R$, the following rounding bounds hold:
\begin{enumerate}[label=(\roman*),nosep]
\item There exists $c\in R$ such that $\|\sigma(t-c)\|_\infty \le \frac{n}{2}$.
\item The same $c$ satisfies $\|\sigma(t-c)\|_2 \le \frac{n}{2}$.
\end{enumerate}
\end{lemma}

\begin{proof}
Since $\{1, \zeta, \dots, \zeta^{n-1}\}$ is a $\Q$-basis of $K$, any element $t \in K\otimes_\Q\R$ admits a unique real-coefficient expansion as
\begin{align}
t=\sum_{m=0}^{n-1} x_m \zeta^m, \ x_m \in \R.
\label{E:exp}
\end{align}
For each $m$, let $c_m=\lfloor x_m \rceil$ denote the nearest integer to $x_m$. Define $c=\sum_{m=0}^{n-1} c_m \zeta^m \in R$. The rounding error coefficients $\epsilon_m=x_m-c_m$ satisfy $|\epsilon_m| \le 1/2$ for all $m$.

(i) For any Galois embedding $\sigma_\ell: K \hookrightarrow \C$ (corresponding to odd $\ell$ modulo $2^k$), we have
\begin{align}
\sigma_\ell(t-c)=\sum_{m=0}^{n-1} \epsilon_m \sigma_\ell(\zeta^m).
\label{E:emb}
\end{align}
Since $\sigma_\ell(\zeta^m)=\zeta^{\ell m}$ is a root of unity, it implies $|\sigma_\ell(\zeta^m)|=1$ for all $m$ and $\ell$. Using the triangle inequality $|x+y|\leq |x|+|y|$, we get 
\begin{align}
|\sigma_\ell(t-c)| \le \sum_{m=0}^{n-1} |\epsilon_m| \cdot |\sigma_\ell(\zeta^m)|=\sum_{m=0}^{n-1} |\epsilon_m|\le \frac{n}{2}
\label{E:sigtri}
\end{align}
as each term is bounded by $1/2$. Since this holds for all embeddings $\sigma_\ell$'s, we get $\|\sigma(t-c)\|_\infty \le n/2$.

(ii) From Lemma \ref{lem:orth}, the embedding vectors $\{\sigma(\zeta^m)\}_{m=0}^{n-1}$ are pairwise orthogonal with squared norm $n$, i.e., 
\begin{align}
\langle \sigma(\zeta^a), \sigma(\zeta^b) \rangle=n\cdot \delta_{a,b}.
\label{E:typair}
\end{align}
We then obtain the squared Euclidean norm as
\begin{align}
\|\sigma(t-c)\|_2^2&=\Big\langle \sum_{m=0}^{n-1} \epsilon_m \sigma(\zeta^m), \sum_{m=0}^{n-1} \epsilon_m \sigma(\zeta^m) \Big\rangle
\notag
  \\
&=\sum_{a,b=0}^{n-1} \epsilon_a \epsilon_b \langle \sigma(\zeta^a), \sigma(\zeta^b) \rangle=n \sum_{m=0}^{n-1} \epsilon_m^2.
\end{align}
Since $|\epsilon_m| \le 1/2$, we get $\epsilon_m^2 \le 1/4$ for all $m$. This further implies that
\begin{align}
\|\sigma(t-c)\|_2^2 \le \frac{n^2}{4}.
\label{E:signorm2}
\end{align}
It yields the bound $\|\sigma(t-c)\|_2 \le n/2$.
\end{proof}

\begin{proposition}\label{prop:coverlower} 
The $L^\infty$-covering radius of the embedded ring $\sigma(R)$ satisfies
\begin{align}
  \frac{\sqrt{n}}{2} \le \mu_\infty(\sigma(R)) \le O(\sqrt{n\log n}).
  \label{E:coverbounds}
\end{align}
In particular, $\sqrt{n}/2\le\mu_\infty(\sigma(R))\le O(\sqrt{n\log n})$, i.e., $\mu_\infty(\sigma(R))=\tilde{\Theta}(\sqrt n)$.

\end{proposition}

\begin{proof}
We first show the lower bound. Consider a target $t=\frac{1}{2}\sum_{m=0}^{n-1}\zeta^m$, for any $c=\sum c_m\zeta^m \in R$ with $c_m\in\Z$, define $\epsilon_m=\frac{1}{2}-c_m$. Since $c_m\in\Z$, each $\epsilon_m$ is a half-integer with $|\epsilon_m|\ge\frac{1}{2}$.

From Lemma \ref{lem:orth}, we have
\begin{align}
  \|\sigma(t-c)\|_2^2=n\sum_{m=0}^{n-1}\epsilon_m^2 \ge \frac{n^2}{4}.
  \label{E:signnorm2}
\end{align}
Since $\|\sigma(t-c)\|_2^2=\sum_{\ell\text{ odd}}|\sigma_\ell(t-c)|^2$ is a sum of $n$ terms, we then get
\begin{align}
  \|\sigma(t-c)\|_\infty^2
  \ge \frac{1}{n}\sum_\ell|\sigma_\ell(t-c)|^2 \ge \frac{n}{4}.
  \label{E:signlow}
\end{align}
So, it yields $\|\sigma(t-c)\|_\infty \ge \sqrt{n}/2$ for every $c\in R$, gives $\mu_\infty(\sigma(R))\ge\sqrt{n}/2$.

In what follows, we prove the upper bound. For $t=\sum_m x_m\zeta^m\in K\otimes_\Q\R$, let $x_m=\lfloor x_m\rceil+\delta_m$ with $\delta_m\in[-1/2,1/2]$, and $c_m=\lfloor x_m\rceil+s_m$, where $s_m\in\{\pm 1, 0\}$ is chosen independently with $\Pr[s_m=+1]=(1+\delta_m)/2$ and $\Pr[s_m=-1]=(1-\delta_m)/2, \Pr[s_m=0]=0$ when $\delta_m\ne 0$, and $s_m\in\{0,\pm 1\}$ uniformly when $\delta_m=0$. Then $\mathbb{E}[c_m-x_m]=0$. So, the embedding-wise error is given by 
\begin{align}
  \sigma_\ell(c-t)=\sum_{m=0}^{n-1}(c_m-x_m)\zeta^{\ell m},
  \label{E:sigma}
\end{align}
which is a sum of independent centered random variables, each bounded in magnitude by $1$. 

By applying the Hoeffding's inequality to the real and imaginary parts of Eq.(\ref{E:sigma}), we obtain 
\begin{align}
  \Pr\bigl[|\sigma_\ell(c-t)|\ge u\bigr] \le 4\exp(-\frac{u^2}{n}).
  \label{E:probles}
\end{align}
Taking a union bound over the $n$ embeddings and setting
$u=\sqrt{n\ln(8n)}$, the right-hand side is less than 1. So, there exists a realization with $\|\sigma(c-t)\|_\infty\le\sqrt{n\ln(8n)}= O(\sqrt{n\log n})$. Since this holds for every $t$, we finally get $\mu_\infty(\sigma(R))=O(\sqrt{n\log n})$.
\end{proof}

\begin{example}\label{ex:covergap}
For $n=256$, the coordinate-wise rounding (Lemma \ref{lem:cover}) guarantees rounding error $\le n/2=128$ in the embedding norm. The true covering radius bound (Proposition \ref{prop:coverlower}) is $\mu_\infty(\sigma(R))\le O(\sqrt{n\log n})\approx 38$, a factor of about $3.4$ tighter. 
\end{example}

\section{Module Reduction via Diagonal CDPR's Analysis}\label{sec:twophase}

In this section we extend the CDPR's attack from ideal lattices ($d=1$) to module lattices ($d\ge 2$) over 2-power cyclotomic rings. Our algorithm inputs an $R$-basis satisfying a balance hypothesis (Definition \ref{D:balanced}) on its $K$-Gram-Schmidt vectors, which is proved for Gaussian distribution, and automatically satisfied for bases sampled from the MLWE distribution (Lemma \ref{lem:mlwebal}).

\paragraph{Main idea.} A rank-$d$ module lattice $\sigma(M)$ has dimension $dn$, which might be too large for CDPR's method (only for ideal lattices). Our idea is that the $K$-linear Gram-Schmidt process can split $M$ into $d$ rank-1 submodules $M_i=M\cap K\tilde{\mathbf{b}}_i$ that are mutually orthogonal after embedding (Remark~\ref{rem:GSembed}). SO, we have the decomposition  $\sigma(M)=\bigoplus_{i=1}^d\sigma(M_i)$ and the covolumes $\det(\sigma(M))=\prod_i\det(\sigma(M_i))$. Moreover, each $M_i$ is a rank-1 submodule that CDPR's method can be applied to get a vector $\mathbf{v}_i$ whose length is governed by $\det(\sigma(M_i))^{1/n}$. This further implies the minimum by their geometric mean $\bigl(\prod_i\det(\sigma(M_i))\bigr)^{1/(dn)}=\det(\sigma(M))^{1/(dn)}$, exactly the module determinant. 

Especially, our algorithm has three steps: (i) compute the $K$-linear Gram--Schmidt vectors $\tilde{\mathbf{b}}_1, \dots, \tilde{\mathbf{b}}_d$; (ii) optionally size-reduce; (iii) apply CDPR's procedure to each $M_i$ and return the shortest candidate, see the summarization in Figure \ref{fig:diagonalcdpr}. 

\begin{figure}[t]
\centering
\begin{tikzpicture}[
  font=\footnotesize,
  >={Latex[length=2mm]},
  box/.style={draw,rounded corners=2pt,inner sep=4pt,align=center,minimum height=8mm},
  modbox/.style={box,fill=blue!8,draw=blue!50},
  linebox/.style={box,fill=orange!10,draw=orange!60,minimum width=14mm},
  cdprbox/.style={box,fill=green!10,draw=green!50!black,minimum width=14mm},
  vbox/.style={box,fill=yellow!15,draw=yellow!50!black,minimum width=14mm},
  finalbox/.style={box,fill=red!10,draw=red!60,minimum width=30mm,minimum height=10mm},
  arrow/.style={->,thick,>=Latex},
  lab/.style={font=\scriptsize\itshape,text=black!70},
]
\node[modbox,minimum width=15mm,minimum height=10mm] (M) at (0,0){module $M\subseteq K^d$\\
   basis $\mathbf{b}_1,\dots,\mathbf{b}_d$\\
   {\scriptsize $\sigma(M)\subset\R^{2dn}$}};
\node[lab,above=2pt of M]{Input};
\node[linebox] (L1) at (3,1.7) {$L_1=K\tilde{\mathbf{b}}_1$};
\node[linebox] (L2) at (3,0.5) {$L_2=K\tilde{\mathbf{b}}_2$};
\node at (3,-0.5) {$\vdots$};
\node[linebox] (Ld) at (3,-1.5) {$L_d=K\tilde{\mathbf{b}}_d$};
\draw[arrow] (M.east) -- node[lab,above,sloped] {$K$-GS} (L1.west);
\draw[arrow] (M.east) -- (L2.west);
\draw[arrow] (M.east) -- (Ld.west);
\node[lab,right=1pt of L1,xshift=10mm,yshift=10mm, text=orange!70!black,align=left] (orth)
  {Lemma~\ref{lem:orth}};
\draw[<-,>=Latex,orange!60,thin] (L1.east) -- (orth);
\draw[<-,>=Latex,orange!60,thin] (L2.east) -- (orth);
\node[cdprbox] (C1) at (6,1.7) {CDPR on $M_1$};
\node[cdprbox] (C2) at (6,0.5) {CDPR on $M_2$};
\node at (6,-0.5) {$\vdots$};
\node[cdprbox] (Cd) at (6,-1.5) {CDPR on $M_d$};
\draw[arrow] (L1) -- (C1);
\draw[arrow] (L2) -- (C2);
\draw[arrow] (Ld) -- (Cd);
\node[vbox] (V1) at (8.5,1.7) {$\mathbf{v}_1$};
\node[vbox] (V2) at (8.5,0.5) {$\mathbf{v}_2$};
\node at (8.5,-0.5) {$\vdots$};
\node[vbox] (Vd) at (8.5,-1.5) {$\mathbf{v}_d$};
\draw[arrow] (C1) -- (V1);
\draw[arrow] (C2) -- (V2);
\draw[arrow] (Cd) -- (Vd);
\node[finalbox] (out) at (12,0)
  {$\mathbf{v}=\arg\min_i\|\sigma(\mathbf{v}_i)\|_2$\\
  {$\|\sigma(\mathbf{v})\|_2\le \sqrt{C}\gamma$\ \ \ \ }
   \\
{\scriptsize \qquad $\cdot \det(\sigma(M))^{1/(dn)}$}};
\draw[arrow] (V1) -- (out);
\draw[arrow] (V2) -- (out);
\draw[arrow] (Vd) -- (out);
\end{tikzpicture}
\caption{Schematic diagonal module reduction (Theorem \ref{T:base}). The rank-$d$ module is split by $K$-Gram-Schmidt into $d$ rank-$1$ submodules that are mutually orthogonal after the Galois embedding (Lemma \ref{lem:orth}), so $\prod_i\det(\sigma(M_i))=\det(\sigma(M))$. CDPR's procedure is applied independently to each $M_i$ and returns a candidate $\mathbf{v}_i$ with length governed by $\det(\sigma(M_i))^{1/n}$. Taking the shortest candidate, the power-mean inequality bounds the minimum by the geometric mean of the per-line bounds, which is exactly the module determinant.}
\label{fig:diagonalcdpr}
\end{figure}

\subsection{Input balance hypothesis}

Our module reduction algorithm inputs an $R$-basis $\mathbf{b}_1, \dots, \mathbf{b}_d$ of the free $R$-module $M$, computes its $K$-Gram-Schmidt vectors $\tilde{\mathbf{b}}_1, \dots, \tilde{\mathbf{b}}_d$, and applies CDPR's procedure to each rank-1 submodule $M_i=M\cap K\tilde{\mathbf{b}}_i$. 

\begin{definition}\label{D:balanced}
  An $R$-basis $\mathbf{b}_1, \dots, \mathbf{b}_d$ of a free $R$-module $M\subseteq K^d$ is $C$-balanced (with $C\ge 1$) if its $K$-Gram-Schmidt vectors satisfy
  \begin{align}
  \|\sigma(\tilde{\mathbf{b}}_i)\|_2^2\le C\cdot n\cdot |\Nm_{K/\Q}(\langle\tilde{\mathbf{b}}_i, \tilde{\mathbf{b}}_i\rangle_K)|^{1/n}, \quad 1\le i\le d.
  \label{E:balance}
  \end{align}
\end{definition}

Condition \eqref{E:balance} requires the trace is close to its AM-GM minimum.

\begin{example}\label{ex:unbalanced}
Consider a rank-$1$ module $M=\beta R\subseteq K$, where $\beta$ has highly skewed conjugates, i.e., $|\sigma_1(\beta)|=c$ large and $|\sigma_\ell(\beta)|=c^{-1/(n-1)}$ for the other $n-1$ embeddings, so $|\Nm(\beta)|=1$. Then $\|\sigma(\beta)\|_2^2\approx a^2$ while $n|\Nm(\langle\beta,\beta\rangle_K)|^{1/n}=n$. The balance constant is $C\approx c^2/n$, which can be arbitrarily large by choosing $c$. This is the kind of input that an MLWE distribution cannot produce.
\end{example}

\begin{remark}
\label{rem:balanceschemes}
The hypothesis \eqref{E:balance} is only for the Galois balance $C$ of Eq.\eqref{E:balance}. For a vector with i.i.d. bounded coordinates the conjugates are asymptotically $\mathcal{CN}(0,n\sigma_c^2)$ (Lemma \ref{lem:mlwebal}), so $C\to e^{\gamma}\approx 1.78$. Galois balance therefore holds for any key whose short generators have small, well-spread coefficients, see Table \ref{Ta:schemes}. The diagonal algorithm of Section \ref{sec:twophase} is special to the power-of-two cyclotomic ring $R=\Z[\zeta_{2^k}]$. Schemes over other rings inherit the balance but not the algorithm. 
\end{remark}

\begin{table}[t]
\centering
\caption{Applicability of the $C$-balance hypothesis \eqref{E:balance}. Ring fit indicates whether the power-of-two cyclotomic of  Sections~\ref{sec:orth}--\ref{sec:milp} applies. Galois balance refers to $C=O(1)$ for the short-generator lines.}
\label{Ta:schemes}
\begin{tabular}{@{}lllll@{}}
\toprule
Scheme & Ring & Short-key law & Ring fit & Galois balance \\
\midrule
ML-KEM   & $\Z[\zeta_{2^k}]$, $n=256$ & CBD$_{\eta}$, $\eta\in\{2,3\}$ & yes & $C=O(1)$
 \\
ML-DSA   & $\Z[\zeta_{2^k}]$, $n=256$ & unif.\ $[-\eta,\eta]$          & yes & $C=O(1)$
 \\
Falcon   & $\Z[\zeta_{2^k}]$, $n\in\{512,1024\}$ & Gaussian $(f,g)$+NTRU & yes & $C=O(1)^{\dagger}$
 \\
Hawk     & $\Z[\zeta_{2^k}]$, $n\in\{512,1024\}$ & short Gaussian basis & yes & $C=O(1)^{\dagger}$
\\
NTRU     & $\Z[x]/(x^{p}{-}1)$, $p$ prime & ternary $(f,g)$ & \emph{no}$^{\ddagger}$ & $C=O(1)$ (heur.)
 \\
BGV/BFV$^{\star\star}$ & $\Z[\zeta_{2^k}]$, $n$ large & ternary/sparse & yes & $C=O(1)^{\S}$
 \\
CKKS$^{\star\star}$    & $\Z[\zeta_{2^k}]$, $n$ large & ternary/sparse & yes & $C=O(1)^{\S}$
 \\
\bottomrule
\multicolumn{5}{@{}p{0.95\linewidth}@{}}{\footnotesize
$^{\dagger}$ The NTRU-type basis is unbalanced between its two lines ($\|\sigma(\tilde{\mathbf{b}}_1)\|\sim\|(f,g)\|$ versus $\|\sigma(\tilde{\mathbf{b}}_2)\|\sim q/\|(f,g)\|$), but this is the inter-line split absorbed by the power mean.
 \newline $^{\ddagger}$ $\Z[x]/(x^{p}{-}1)$ is not a power-of-two cyclotomic ring (and not a domain); the ternary secret is Galois-balanced.
 \newline $^{\S}$ Holds for dense ternary/Gaussian secrets; for sparse secrets of fixed Hamming weight $h$ the bound is governed by $h$ and the geometric-mean term can degrade. The large RNS modulus $q$ does not affect $C$, consistent with the $q$-independence of $\sigma_{g_0}$ (Part III).}
\end{tabular}
\end{table}

\begin{lemma}\label{lem:mlwebal}
Let $\mathbf{b}_1, \dots, \mathbf{b}_d\in R^d$ be chosen independently with coefficients drawn i.i.d. from a bounded real distribution with mean zero, variance $\sigma_c^2>0$, and finite fourth moment. Let   $B_i:=\langle\tilde{\mathbf{b}}_i, \tilde{\mathbf{b}}_i\rangle_K\in K^+$ for the $K$-valued squared length of the $i$-th Gram-Schmidt vector, and set $W^{(i)}_\ell:=\tfrac{|\sigma_\ell(\tilde{\mathbf{b}}_i)|^2}{n\sigma_c^2}$ with $\ell\in J:=\{1, 3, \dots, n-1\}$. Then the following results hold
  \begin{enumerate}[label=(\roman*),nosep]
    \item The balance ratio of line $i$ equals the
          arithmetic-to-geometric-mean ratio
    \begin{align}
      \frac{\|\sigma(\tilde{\mathbf{b}}_i)\|_2^2}{n|\Nm_{K/\Q}(B_i)|^{1/n}}=\frac{\operatorname{AM}(W^{(i)})}{\operatorname{GM}(W^{(i)})} \ge 1.
      \label{E:amgclt}
    \end{align}
    
    \item As $n\to\infty$, the per-line ratio converges in probability:
    \begin{align}
      \frac{\operatorname{AM}(W^{(i)})}{\operatorname{GM}(W^{(i)})}
      \xrightarrow{p}\frac{d-i+1}{e^{\psi(d-i+1)}},
      \label{E:perlinlim}
    \end{align}
    where $\psi$ denotes the digamma function;
    
    \item The balance constant $C=\max_{1\le i\le d}\operatorname{AM}(W^{(i)})/\operatorname{GM}(W^{(i)})$ satisfies $C\xrightarrow{p}e^{\gamma}=1.7811\ldots$, where $\gamma$ is     the Euler-Mascheroni constant. In particular, for every $\epsilon>0$ the basis is $(e^{\gamma}+\epsilon)$-balanced with probability $1-o_n(1)$, uniformly in $d$.
  \end{enumerate}
  
\end{lemma}

\begin{proof}
We prove in four steps.

\textit{Step 1 Proof of (i).} Note $B_i=\langle\tilde{\mathbf{b}}_i, \tilde{\mathbf{b}}_i\rangle_K =\sum_{m=1}^d\tilde b_{im}\overline{\tilde b_{im}}\in K^+$, so each embedding satisfies $\sigma_\ell(B_i)=|\sigma_\ell(\tilde{\mathbf{b}}_i)|^2\ge0$. For the $2^k$-th cyclotomic field, by definition of the trace norm and conjugate symmetry, we get 
\begin{align}
  \|\sigma(\tilde{\mathbf{b}}_i)\|_2^2
  &=\sum_{\ell\text{ odd}}|\sigma_\ell(\tilde{\mathbf{b}}_i)|^2
   =\sum_{\ell\text{ odd}}\sigma_\ell(B_i)
  \notag\\
  &=2\sum_{\ell\in J}\sigma_\ell(B_i)
   =2 n\sigma_c^2\sum_{\ell\in J}W^{(i)}_\ell
  \notag\\
  & =n^2\sigma_c^2\cdot\operatorname{AM}(W^{(i)}),
  \label{E:nums}
\end{align}
where the last equality has used the equality of $\operatorname{AM}(W^{(i)})=\tfrac{1}{|J|}\sum_{\ell\in J}W^{(i)}_\ell=\tfrac{2}{n}\sum_{\ell\in J}W^{(i)}_\ell$.

Note the field norm is the product over all $n$ embeddings. This implies 
\begin{align}
  \Nm_{K/\Q}(B_i)
  &=\prod_{\ell\text{ odd}}\sigma_\ell(B_i)
   =\prod_{\ell\in J}\sigma_\ell(B_i)^2
  \notag\\
  &=\prod_{\ell\in J}\bigl(n\sigma_c^2  W^{(i)}_\ell\bigr)^2
  =(n\sigma_c^2)^{n} \Bigl(\prod_{\ell\in J}W^{(i)}_\ell\Bigr)^{2}.
  \label{E:norms}
\end{align}
Taking the $n$-th root and multiplying by $n$, we get from Eq.(\ref{E:norms})
\begin{align}
  n|\Nm_{K/\Q}(B_i)|^{1/n}
  =n^2\sigma_c^2\Bigl(\prod_{\ell\in J}W^{(i)}_\ell\Bigr)^{2/n}
  =n^2\sigma_c^2\cdot\operatorname{GM}(W^{(i)}),
  \label{E:denoms}
\end{align}
where $\operatorname{GM}(W^{(i)})
=\bigl(\prod_{\ell\in J}W^{(i)}_\ell\bigr)^{2/n}$ is the geometric mean of $|J|=n/2$ values. Dividing Eq.\eqref{E:nums} by Eq.\eqref{E:denoms}, we obtain Eq.\eqref{E:amgclt}, where the inequality $\ge 1$ is from the classical AM--GM inequality.

\medskip
\textit{Step 2 Asymptotic distribution of $W^{(i)}_\ell$.} Fix an odd residue $\ell$. Each ring-element coefficient $\sigma_\ell(b_{jm})=\sum_{t=0}^{n-1}(b_{jm})_t\zeta^{\ell t}$ is a linear combination of $n$ independent, mean-zero, bounded random variables (the power-basis coefficients $(b_{jm})_t$) with deterministic unimodular weights $\zeta^{\ell t}$. Let $U_\ell:=\Re \sigma_\ell(b_{jm})/\sqrt{n\sigma_c^2}$ and $V_\ell:=\Im\sigma_\ell(b_{jm})/\sqrt{n\sigma_c^2}$, similar to Lemma \ref{lem:orth}, we can obtain
\begin{align}
  \E[U_\ell^2]=\E[V_\ell^2]=\tfrac{1}{2}, \ \ \E[U_\ell V_\ell]=0.
  \label{E:cltmoments}
\end{align}
By the Lindeberg Center Limit Theorem (CLT) (see Part III for full verification), $(U_\ell,V_\ell)\xrightarrow{d}\mathcal{N}(\mathbf{0},\tfrac{1}{2}I_2)$ jointly, so 
$\sigma_\ell(b_{jm})/\sqrt{n\sigma_c^2}\xrightarrow{d}\mathcal{CN}(0,1)$. Moreover, the convergence is joint over the $d^2$ entries $(j,m)$, which are independent because they involve disjoint sets of input coefficients.

Consider the $d\times d$ matrix
\begin{align}
  A_\ell:=\frac{1}{\sqrt{n\sigma_c^2}}
  \bigl[\sigma_\ell(\mathbf{b}_1) \big|\dots 
  \big|\sigma_\ell(\mathbf{b}_d)\bigr]\in\C^{d\times d}
  \label{E:ginibrelim}
\end{align}
it converges in distribution to a standard complex Ginibre matrix (i.i.d. $\mathcal{CN}(0,1)$ entries) \cite{Goodman63,Edelman89}.

As each embedding $\sigma_\ell: K\to\C$ is a ring homomorphism and commutes with complex conjugation on $K$, the $K$-inner product embeds faithfully: $\sigma_\ell(\langle\mathbf{a},\mathbf{b}\rangle_K)=\langle\sigma_\ell(\mathbf{a}), \sigma_\ell(\mathbf{b})\rangle_{\C^d}$. So, the $K$-Gram-Schmidt recurrence \eqref{E:GSKO}, evaluated at embedding $\ell$, is the ordinary complex Gram-Schmidt of the columns of $A_\ell$:
\begin{align}
  \sigma_\ell(\tilde{\mathbf{b}}_j)
    =\sigma_\ell(\mathbf{b}_j)-\sum_{i<j}\frac{\langle\sigma_\ell(\mathbf{b}_j),                    \sigma_\ell(\tilde{\mathbf{b}}_i)\rangle_{\C^d}}                  {\|\sigma_\ell(\tilde{\mathbf{b}}_i)\|_{\C^d}^2}
   \sigma_\ell(\tilde{\mathbf{b}}_i).
  \label{E:gsembed}
\end{align}
Since the complex Gram-Schmidt is a continuous function of the input matrix, using the continuous mapping theorem and Eq.\eqref{E:ginibrelim}, we get 
\begin{align}
W^{(i)}_\ell=\frac{|\sigma_\ell(\tilde{\mathbf{b}}_i)|^2}{n\sigma_c^2}
\xrightarrow{d}|R_{ii}|^2,
\label{E:Weg}
\end{align}
where $R_{ii}$ is the $i$-th diagonal entry of the upper-triangular factor in the QR decomposition of a Ginibre matrix. By the complex Bartlett decomposition \cite{Goodman63,Edelman89} (see also \cite{Muirhead82,Forrester10}), we have $|R_{ii}|^2\sim\Gamma(d-i+1,1)$, with the $d$ diagonal entries mutually independent. In particular, the first line has the richest distribution $\Gamma(d,1)$, and the last line the sparsest $\Gamma(1,1)=\operatorname{Exp}(1)$.

\medskip
\textit{Step 3.} For distinct $\ell,\ell'\in J$, similar to Lemma \ref{lem:orth} with $c=\ell-\ell'$, we have the cross-covariance vanishes, i.e., 
\begin{align}
\E[\sigma_\ell(b)\overline{\sigma_{\ell'}(b)}]=\sigma_c^2\sum_{t=0}^{n-1}\zeta^{(\ell-\ell')t}=0,
\label{E:covva}
\end{align}
where the geometric sum vanishes because $\ell\ne\ell'$ are distinct odd residues and $\zeta^{\ell-\ell'}\ne1$. 

Combined with the joint CLT of Step 2, the matrices $\{A_\ell\}_{\ell\in J}$ are asymptotically independent. So, the family $\{W^{(i)}_\ell\}_{\ell\in J}$ is asymptotically i.i.d. with marginal law $\Gamma(m_i,1)$, and $m_i:=d-i+1$.

Since $\Gamma(m,1)$ has finite mean $m$ and finite variance $m$, the weak law of large numbers gives
\begin{align}
  \operatorname{AM}(W^{(i)})=\frac{1}{|J|}\sum_{\ell\in J}W^{(i)}_\ell \xrightarrow{p}\E[\Gamma(m_i,1)]=m_i.
  \label{E:amlln}
\end{align}
For the geometric mean, let $\log\operatorname{GM}(W^{(i)})=\frac{1}{|J|}\sum_{\ell\in J}\log W^{(i)}_\ell$.

The random variable $\log\Gamma(m,1)$ has mean $\psi(m)$ and variance $\psi'(m)<\infty$ (the trigamma function), so the same law of large numbers yields
\begin{align}
  \log\operatorname{GM}(W^{(i)})\xrightarrow{p}\psi(m_i),
    \label{E:gmlln}
\end{align}
hence, we get $\operatorname{GM}(W^{(i)})\xrightarrow{p}e^{\psi(m_i)}$. 

Combining Eqs.\eqref{E:amlln} and \eqref{E:gmlln}, using the continuous mapping theorem we obtain
\begin{align}
  \frac{\operatorname{AM}(W^{(i)})}{\operatorname{GM}(W^{(i)})}
  \xrightarrow{p}\frac{m_i}{e^{\psi(m_i)}}
  =\frac{d-i+1}{e^{\psi(d-i+1)}}.
  \label{E:ratiolim}
\end{align}

\medskip
\textit{Step 4 Proof of (iii).} Consider the function $f(m)=m e^{-\psi(m)}$ on $[1,\infty)$. We show $f$ is strictly decreasing. Indeed, we have $(\log f)'(m)=\tfrac{1}{m}-\psi'(m)$. Since $\psi'(m)=\sum_{j=0}^{\infty}(m+j)^{-2}$, we have
\begin{align}
  \psi'(m)=\sum_{j=0}^{\infty}\frac{1}{(m+j)^2}>\int_0^{\infty}\frac{dx}{(m+x)^2}=\frac{1}{m}.
  \label{E:psiineq}
\end{align}
The inequality is strict because the sum dominates the integral. So, $(\log f)'(m)<0$, i.e., $f$ is strictly decreasing on $[1,\infty)$.

The GS lines $i=1, \dots, d$ have shape parameters $m_i=d, d-1, \dots, 1$. Since $f$ is decreasing function, using the identity $\psi(1)=-\gamma$, we have 
\begin{align}
  \max_{1\le i\le d}\frac{m_i}{e^{\psi(m_i)}}=f(1)=\frac{1}{e^{\psi(1)}}=e^{-\psi(1)}=e^{\gamma},
  \label{E:maxline}
\end{align}
where $\gamma=0.5772\ldots$ is the Euler-Mascheroni constant. Hence, $e^{\gamma}=1.7811\ldots$.

Since each of the $d$ ratios converges in probability in Eq.\eqref{E:ratiolim}, the maximum of finitely many convergent sequences also converges as
\begin{align}
  C=\max_{1\le i\le d}
    \frac{\operatorname{AM}(W^{(i)})}{\operatorname{GM}(W^{(i)})}  \xrightarrow{p}\max_{1\le i\le d}\frac{m_i}{e^{\psi(m_i)}}
    =e^{\gamma}.
  \label{E:Climit}
\end{align}

By the definition of convergence in probability, for every $\epsilon>0$,
$\Pr[C>e^{\gamma}+\epsilon]\to0$ as $n\to\infty$, so the basis is
$(e^{\gamma}+\epsilon)$-balanced with probability $1-o_n(1)$. Since
$e^{\gamma}$ is independent of $d$, the statement holds uniformly in $d$.
\end{proof}

Now, we consider the balance hypothesis \ref{D:balanced} with special distribution as follows. 

\begin{theorem}
\label{T:gaussbalance}
Let $\mathbf{b}_1, \dots, \mathbf{b}_d\in (K\otimes_\Q\R)^d$ have coordinates whose power-basis coefficients are i.i.d. $\mathcal{N}(0,\sigma_c^2)$. Let $J=\{1, 3, \dots, n-1\}$ and $\tilde{\mathbf{b}}_1, \dots, \tilde{\mathbf{b}}_d$ be the $K$-Gram-Schmidt vectors. Then the following results hold 
\begin{enumerate}[label=(\roman*),nosep]
\item For each line $i$ the normalized squared embedded GS-lengths
$W^{(i)}_\ell:=\tfrac{|\sigma_\ell(\tilde{\mathbf{b}}_i)|^2}{n\sigma_c^2}$ with $\ell\in J$ are i.i.d. across $\ell\in J$ with law $\Gamma(d-i+1, 1)$;

\item The balance ratio of line $i$ equals exactly the arithmetic-to-geometric-mean ratio of these variables as
\begin{align}
  \frac{\|\sigma(\tilde{\mathbf{b}}_i)\|_2^2}{n|\Nm_{K/\Q}(\langle\tilde{\mathbf{b}}_i, \tilde{\mathbf{b}}_i\rangle_K)|^{1/n}}
=\frac{\operatorname{AM}(W^{(i)})}{\operatorname{GM}(W^{(i)})} =\frac{\tfrac1{|J|}\sum_{\ell\in J}W^{(i)}_\ell}        {\bigl(\prod_{\ell\in J}W^{(i)}_\ell\bigr)^{1/|J|}} \ge 1 .
   \label{E:amgm}
\end{align}
\item $C=\max_{1\le i\le d}\operatorname{AM}(W^{(i)})/\operatorname{GM}(W^{(i)})$ converges almost surely to $e^{\gamma}=1.7811\ldots$ as $n\to\infty$, and for every $\delta>0$, 
\begin{align}
  \Pr\bigl[C> e^{\gamma}(1+\delta)\,\bigr]\le 4d e^{-c(\delta) n}
\end{align}
for a constant $c(\delta)>0$. In particular, a Gaussian basis is $(e^{\gamma}+\delta)$-balanced with probability $1-e^{-\Omega(n)}$, uniformly in $d$.
\end{enumerate}

\end{theorem}

\begin{proof}
Given an odd $\ell$, $\sigma_\ell(b)=\sum_{i=0}^{n-1}c_i\zeta^{\ell i}$ is a real-linear image of the Gaussian vector $(c_i)$, hence jointly Gaussian. Using $\sum_{i=0}^{n-1}\zeta^{2\ell i}=0$ (a geometric sum with ratio $\zeta^{2\ell}\ne1$ and $(\zeta^{2\ell})^{n}=1$), we obtain 
\begin{align}    
 &\E[\Re\sigma_\ell(b)^2]=\E[\Im\sigma_\ell(b)^2]=\tfrac{n\sigma_c^2}{2}
\\ 
 &\E[\Re\sigma_\ell(b)\,\Im\sigma_\ell(b)]=0
\end{align}
So, we get $\sigma_\ell(b)\sim\mathcal{CN}(0,n\sigma_c^2)$. For different $\ell, \ell'\in J$, we have 
\begin{align}   
 \E[\sigma_\ell(b)\overline{\sigma_{\ell'}(b)}] =\sigma_c^2\sum_i\zeta^{(\ell-\ell')i}=0,
  \label{E:exporth}
\end{align}
i.e., jointly Gaussian and Hermitian-uncorrelated with circular symmetry implies independence. As the entries $b_{jm}$ are independent ring elements, we know the matrices $A_\ell:=[\sigma_\ell(\mathbf{b}_1)| \cdots| \sigma_\ell(\mathbf{b}_d)] \in\C^{d\times d}$ (dividing $\sqrt{n\sigma_c^2}$) are independent complex Ginibre matrices \cite{Muirhead82,Forrester10}.

Since each $\sigma_\ell$ is a ring homomorphism commuting with the $K$-Hermitian form, $\sigma_\ell(\langle\mathbf{a}, \mathbf{b}\rangle_K)=\langle\sigma_\ell(\mathbf{a}), \sigma_\ell(\mathbf{b})\rangle_{\C^d}$, so the $K$-Gram-Schmidt of $(\mathbf{b}_j)$ embeds at $\ell$ to the ordinary complex Gram-Schmidt of the columns of $A_\ell$. By the QR/Bartlett decomposition of a complex Ginibre matrix \cite{Goodman63,Edelman89}, the squared GS-lengths $|R^{(\ell)}_{ii}|^2$ are independent with $|R^{(\ell)}_{ii}|^2/(n\sigma_c^2)\sim\Gamma(d-i+1,1)$. This proves the part (i), where the independence across $\ell\in J$ is from that of the $A_\ell$.

By conjugate symmetry, we have 
\begin{align}
\|\sigma(\tilde{\mathbf{b}}_i)\|_2^2
=2\sum_{\ell\in J}|\sigma_\ell(\tilde{\mathbf{b}}_i)|^2
=2n\sigma_c^2\sum_{\ell\in J}W^{(i)}_\ell
\end{align}
Moreover, we have 
\begin{align}
\Nm_{K/\Q}(\langle\tilde{\mathbf{b}}_i,\tilde{\mathbf{b}}_i\rangle_K)
=\prod_{\text{odd}\ \ell}|\sigma_\ell(\tilde{\mathbf{b}}_i)|^2
=(n\sigma_c^2)^n\prod_{\ell\in J}\bigl(W^{(i)}_\ell\bigr)^2    
\end{align}
So, we get
\begin{align}
n|\Nm|^{1/n}=n^2\sigma_c^2(\prod_{\ell\in J}W^{(i)}_\ell)^{2/n}
\end{align}
This further implies Eq.\eqref{E:amgm}, where the inequality $\ge 1$ is from AM-GM inequality.

Note the variables $W\sim\Gamma(m,1)$ and $\log W$ are sub-exponential, with $\E[W]=m$, $\E[\log W]=\psi(m)$. Using the Bernstein's inequality \cite{Billingsley95}, for each $i$ and small $\epsilon>0$, we have 
\begin{align}
&\Pr[|\operatorname{AM}(W^{(i)})-m_i|>\epsilon]\le2e^{-c_1\epsilon^2 n},
\\
&\Pr[|\log\operatorname{GM}(W^{(i)})-\psi(m_i)|>\epsilon]\le2e^{-c_2\epsilon^2 n}    
\end{align}
where $m_i=d-i+1$. Off the union of these $2d$ events, we get 
\begin{align}
 \frac{\operatorname{AM}(W^{(i)})}{\operatorname{GM}(W^{(i)})}
\le m_i e^{-\psi(m_i)}(1+\epsilon)e^{\epsilon}   
\end{align}
for every $i$. 

The map $m\mapsto m e^{-\psi(m)}$ is strictly decreasing on $[1,\infty)$, as $\tfrac{d}{dm}(\log m-\psi(m))=\tfrac{1}{m}-\psi'(m)<0$, where $\psi'(m)=\sum_{j\ge0}(m+j)^{-2}>\int_0^\infty(m+x)^{-2}dx=\tfrac{1}{m}$. Hence, its maximum over $i\in\{1, \dots, d\}$ is achieved $e^{-\psi(1)}=e^{\gamma}$ by $m_i=1$ (the last line $i=d$). Choosing $\epsilon$ with $(1+\epsilon)e^{\epsilon}\le 1+\delta$, we obtain part (iii), where the a.s. limit is from the strong law applied to $\operatorname{AM}$ and $\log\operatorname{GM}$.
\end{proof}

\begin{remark}
Theorem \ref{T:gaussbalance} gives the strict result of Lemma \ref{lem:mlwebal} in Gaussian case: the convergence is to the exact constant $e^\gamma$ with exponential concentration, and the binding line is the last Gram-Schmidt vector (Gamma shape $1$). The value $d/e^{\psi(d)}$ in Lemma \ref{lem:mlwebal} is the ratio of the first line (shape $d$). Since $m\mapsto m e^{-\psi(m)}$ is decreasing, the true balance constant is $\max_i m_i e^{-\psi(m_i)}=e^{\gamma}$, independent of $d$. There remains the same conclusion $C=O(1)$ uniformly in $d$, and the module reduction factor $\alpha_d=\sqrt{C}\leq 1.3346<1.4$.
\end{remark}

\subsection{$R$-Linear Gram–Schmidt orthogonalization}\label{subsec:rgs}

Let $M$ be a free $R$-module of rank $d$ with basis $\mathbf{b}_1, \dots, \mathbf{b}_d \in K^d$. We perform a $K$-linear Gram-Schmidt orthogonalization using the Hermitian inner product on $K^d$ as 
\begin{align}
  \tilde{\mathbf{b}}_j=\mathbf{b}_j-\sum_{i=1}^{j-1} \mu_{ji} \tilde{\mathbf{b}}_i,
  \label{E:GSKO}
\end{align}
where Gram-Schmidt coefficients $\mu_{ji}=\langle \mathbf{b}_j, \tilde{\mathbf{b}}_i \rangle_K/\langle \tilde{\mathbf{b}}_i, \tilde{\mathbf{b}}_i \rangle_K \in K$.

The Gram–Schmidt vectors are pairwise orthogonal under the $K$-inner product, i.e., $\langle \tilde{\mathbf{b}}_i, \tilde{\mathbf{b}}_j \rangle_K=0$ for $i \neq j$. Since the basis $\{\mathbf{b}_i\}$ is linearly independent over $K$, each $\tilde{\mathbf{b}}_i$ is non-zero. Hence, we have $\langle \tilde{\mathbf{b}}_i, \tilde{\mathbf{b}}_i \rangle_K \in K^+ \setminus \{0\}$.

The module determinant as the determinant of the $K$-valued Gram matrix is given by 
\begin{align}
{\rm det}_R(M)=\det(\langle \mathbf{b}_i, \mathbf{b}_j \rangle_K)_{1\le i, j\le d} \in K^\times /R^\times.
\label{E:detR}
\end{align}
From the orthogonality of the Gram–Schmidt vectors, we have
\begin{align}
  {\rm det}_R(M)=\prod_{i=1}^d \langle \tilde{\mathbf{b}}_i, \tilde{\mathbf{b}}_i \rangle_K.
  \label{E:detprod}
\end{align}

\begin{remark}\label{rem:GSembed}
The coordinate-wise Galois embedding $\sigma:K^d\to\C^{dn}$ maps the $K$-valued inner product to the standard Hermitian inner product on $\C^{dn}$ as 
  \begin{align}    
  \langle\sigma(\mathbf{a}_i), \sigma(\mathbf{b}_j)\rangle_{\C^{dn}}=\Tr_{K/\Q}(\langle\mathbf{a}_i,\mathbf{b}_j\rangle_K).
  \label{E:GSembed}
  \end{align}
  In particular, if $\langle\mathbf{a}_i, \mathbf{b}_j\rangle_K=0$ then $\sigma(\mathbf{a}_i)$ and $\sigma(\mathbf{b}_j)$ are orthogonal in $\C^{dn}$. For the scalar multiplication by $\alpha\in K$, we have the embedding-wise factorization
  \begin{align}    \|\sigma(\mathbf{a}_i\cdot\alpha)\|_2^2=\sum_{\ell}|\sigma_\ell(\alpha)|^2\sum_{k=1}^{d}|\sigma_\ell(a_{ik})|^2,
    \label{E:scalarnorm}
  \end{align}
  which factorizes as $\|\sigma(\mathbf{a}_i)\|_2^2 \|\sigma(\alpha)\|_2^2/n$ when $|\sigma_\ell(\alpha)|$ is independent of $\ell$.
\end{remark}

\subsection{Optional size reduction}
\label{subsec:phase1}

In this section we show the optional size reduction, which does not affect the approximation-factor bound in Theorem \ref{T:base}, but improves numerical conditioning. For each $j$ from $2$ to $d$ and each $i$ from $j-1$ down to $1$, we compute
\begin{align}
  c_{ji}=\Round(\mu_{ji})\in R, \mathbf{b}_j \gets \mathbf{b}_j-c_{ji}\cdot\mathbf{b}_i,
  \label{E:sizered}
\end{align}
where $\Round:K\to R$ is a rounding function whose error is controlled by coordinate-wise rounding (Lemma \ref{lem:cover}) with $\rho=n/2$ or Chinese Remainder Theorem (CRT)-scaled rounding (Section \ref{sec:crt}) with $\rho\le 1$. 

\subsection{Diagonal CDPR's algorithm on rank-1 submodules}

Note each Gram-Schmidt vector $\tilde{\mathbf{b}}_i$ spans a one-dimensional $K$-subspace $L_i:= K\cdot\tilde{\mathbf{b}}_i \subseteq K^d$. The intersection $M_i:= M\cap L_i$ is a non-zero $R$-submodule of $M$ of $R$-rank $1$. Under the $R$-linear isomorphism $\psi_i: L_i\xrightarrow{\sim} K$, $\alpha\tilde{\mathbf{b}}_i\mapsto\alpha$, the image $J_i:= \psi_i(M_i)\subseteq K$ is a non-zero fractional $R$-ideal, and $M_i=J_i\cdot\tilde{\mathbf{b}}_i$.

\bigskip
\noindent\textit{Applying CDPR's algorithm.} For $k\le 8$, the ring $R$ is a PID, so $J_i$ is principal. For $k\ge 9$, the free $R$-modules arising from MLWE generically yield principal $J_i$. Here, the non-principal case will be deal with the algorithm of Biasse-Song \cite{BiasseSong16} at the same asymptotic cost. Now, pick any generator $\alpha_{i,0}\in J_i$ and apply the CDPR's short-generator algorithm \cite{CDPR16,CDW17} to the principal ideal $(\alpha_{i,0})\subseteq R$. This returns a unit $u_i\in R^\times$ such that $\alpha_i':= u_i^{-1}\alpha_{i,0}$ is a generator of $J_i$ satisfying
\begin{align}
  \|\sigma(\alpha_i')\|_\infty \le \gamma
  |\Nm_{K/\Q}(J_i)|^{1/n},
  \label{E:cdprbound}
\end{align}
where $\gamma=\exp(C_0\sqrt n\log n)$ for a constant $C_0>0$. 

\bigskip
\noindent\textit{The output vector.} Define the output vector as 
\begin{align}
  \mathbf{v}_i:= \alpha_i'\cdot\tilde{\mathbf{b}}_i
  \in M_i \subseteq M.
  \label{E:videf}
\end{align}
By construction $\mathbf{v}_i$ is a non-zero element of $M$.

\bigskip
\noindent\textit{Per-line norm bound.} We use the inequality $\|\sigma(\mathbf{v}_i)\|_2\le\sqrt{dn}\|\sigma(\mathbf{v}_i)\|_\infty$ because $\sigma(\mathbf{v}_i)\in\C^{dn}$ has at most $dn$ complex coordinates. In what follows, we get tight factorization via the trace as 
\begin{align}
  \|\sigma(\mathbf{v}_i)\|_2^2
  &= \sum_\ell|\sigma_\ell(\alpha_i')|^2\cdot     \sigma_\ell(\langle\tilde{\mathbf{b}}_i,\tilde{\mathbf{b}}_i\rangle_K)
  \notag\\
  &\le \|\sigma(\alpha_i')\|_\infty^2 \Tr_{K/\Q}(\langle\tilde{\mathbf{b}}_i, \tilde{\mathbf{b}}_i\rangle_K)
  \notag\\
  &= \|\sigma(\alpha_i')\|_\infty^2\ 
     \|\sigma(\tilde{\mathbf{b}}_i)\|_2^2
  \notag\\
  &\le \gamma^2|\Nm_{K/\Q}(J_i)|^{2/n} \ 
     \|\sigma(\tilde{\mathbf{b}}_i)\|_2^2.
  \label{E:viprel}
\end{align}

Under the $C$-balance hypothesis \eqref{E:balance} (Definition \ref{D:balanced}), we obtain 
\begin{align}
  \|\sigma(\tilde{\mathbf{b}}_i)\|_2^2
  \le C n |\Nm_{K/\Q}(\langle\tilde{\mathbf{b}}_i, \tilde{\mathbf{b}}_i\rangle_K)|^{1/n}.
  \label{E:BiB}
\end{align}
From the inequalities \eqref{E:BiB} and \eqref{E:viprel}, we then obtain 
\begin{align}
  \|\sigma(\mathbf{v}_i)\|_2^2 \le
  C n \gamma^2 |\Nm_{K/\Q}(J_i)|^{2/n} |\Nm_{K/\Q}(\langle\tilde{\mathbf{b}}_i, \tilde{\mathbf{b}}_i\rangle_K)|^{1/n}.
  \label{E:viboundint}
\end{align}

\bigskip
\noindent\textit{Covolume of $\sigma(M_i)$ and product identity.} The rank-1 sublattice $M_i=J_i\tilde{\mathbf{b}}_i$ can be embeded to a rank-$n$ lattice $\sigma(M_i)\subset\sigma(L_i)$ with covolume
\begin{align}
  \det(\sigma(M_i))=|\Nm_{K/\Q}(J_i)|\cdot
  n^{n/2}\cdot|\Nm_{K/\Q}(\langle\tilde{\mathbf{b}}_i,\tilde{\mathbf{b}}_i\rangle_K)|^{1/2},
  \label{E:covolMi}
\end{align}
which is from the same Vandermonde computation as Eq.\eqref{E:detsigmaM}. Taking the $(2/n)$-th power, we get 
\begin{align}
  |\Nm_{K/\Q}(J_i)|^{2/n}\, |\Nm_{K/\Q}(\langle\tilde{\mathbf{b}}_i, \tilde{\mathbf{b}}_i\rangle_K)|^{1/n}=\frac{1}{n}\det(\sigma(M_i))^{2/n}.
  \label{E:Miexp}
\end{align}

Combining Eq.\eqref{E:Miexp} and inequality \eqref{E:viboundint} give 
\begin{align}
  \|\sigma(\mathbf{v}_i)\|_2^2 \le
  C\gamma^2\det(\sigma(M_i))^{2/n}.
  \label{E:vibound}
\end{align}

Since the $K$-lines $L_i$ are pairwise orthogonal under the $K$-inner product (Remark \ref{rem:GSembed}), the embedded sublattices $\sigma(M_i)$ are mutually orthogonal in $\R^{2dn}$, and $\sigma(M)=\bigoplus_{i=1}^d\sigma(M_i)$ as orthogonal direct sum. Therefore, we finally get 
\begin{align}
  \prod_{i=1}^d \det(\sigma(M_i))=\det(\sigma(M)).
  \label{E:Miprod}
\end{align}
Here, the $d$ submodules are applied independently since the $L_i$ are pairwise orthogonal.

\subsection{Base module reduction bound}

We now state the main bound for the base algorithm. 

\begin{theorem}\label{T:base}
  Let $M\subseteq K^d$ be a free $R$-module of rank $d$. Given by a $C$-balanced $R$-basis $\mathbf{b}_1, \dots, \mathbf{b}_d$, consider the following algorithm: (i) computes the $K$-Gram-Schmidt decomposition   $\tilde{\mathbf{b}}_1,\cdots,\tilde{\mathbf{b}}_d$, (ii) applies CDPR's algorithm to each rank-1 submodule $M_i=M\cap(K\tilde{\mathbf{b}}_i)$ to obtain a short generator $\alpha_i'\in J_i$, and (iii) returns 
  \begin{align}
    \mathbf{v}=\arg\min_{1\le i\le d}\|\sigma(\mathbf{v}_i)\|_2,
    \quad \mathbf{v}_i:=\alpha_i'\cdot\tilde{\mathbf{b}}_i.
    \label{E:outpurul}
  \end{align}
  This algorithm produces a non-zero vector $\mathbf{v}\in M$ with
  \begin{align}
    \|\sigma(\mathbf{v})\|_2 \le
    \sqrt{C} \gamma \det(\sigma(M))^{1/(dn)},
    \label{E:baseB}
  \end{align}
  where $\gamma=\exp(C_0\sqrt n\log n)$ is the CDPR factor. 
\end{theorem}

Comparing with the Hermite upper bound   $\lambda_1(\sigma(M))\le\sqrt{dn}\cdot\det(\sigma(M))^{1/(dn)}$, the algorithm in Theorem \ref{T:base} achieves Hermite factor at most  $\sqrt{C/(dn)}\gamma=\exp(\tO(\sqrt{n}))$. Under the Gaussian heuristic $\lambda_1(\sigma(M))\approx\sqrt{dn/(2\pi e)}\cdot\det(\sigma(M))^{1/(dn)}$, the Module-SVP approximation factor is $\gamma=O(\sqrt{C})\cdot\exp(\tO(\sqrt n))$. For input bases sampled from the MLWE distribution, we have $C=O(1)$ with probability $1-o(1)$ (Lemma \ref{lem:mlwebal}), so, $\gamma=\exp(\tO(\sqrt{n}))$ unconditionally.

\begin{proof}
Set $A_i:=\|\sigma(\mathbf{v}_i)\|_2^2$. By the inequality \eqref{E:vibound}, we get 
\begin{align}
  A_i \le C\gamma^2\det(\sigma(M_i))^{2/n}.
  \label{E:Aiboundn}
\end{align}

Using the output in Eq.\eqref{E:outpurul}, we have $\|\sigma(\mathbf{v})\|_2^2=\min_{1\le i\le d} A_i$. From the power-mean inequality we get
\begin{align}
  \min_{1\le i\le d} A_i \le
  \Bigl(\prod_{i=1}^d A_i\Bigr)^{1/d}.
  \label{E:powermean}
\end{align}

Combining inequalities \eqref{E:Aiboundn} and \eqref{E:powermean}, and the orthogonal direct-sum identity \eqref{E:Miprod}, we obtain 
\begin{align}
  \|\sigma(\mathbf{v})\|_2^2
  &\le C\gamma^2\Bigl(\prod_{i=1}^d\det(\sigma(M_i))\Bigr)^{2/(dn)} \notag\\
  &=C\gamma^2\det(\sigma(M))^{2/(dn)}.
  \label{E:afterpm}
\end{align}
This further yields the inequality \eqref{E:baseB}, where the Hermite factor and heuristic Module-SVP factor are followed directly.
\end{proof}

The bound \eqref{E:baseB} requires the $C$-balance hypothesis (Definition \ref{D:balanced}) on the input basis. This holds for MLWE-sampled bases (Lemma \ref{lem:mlwebal}), and more generally for any  basis where the Galois conjugates of the  $\langle\tilde{\mathbf{b}}_i, \tilde{\mathbf{b}}_i\rangle_K$ are approximately equidistributed. An worst-case input basis may violate the hypothesis, in that case, the basis can be preprocessed to satisfy balance via algebraic LLL \cite{LLPS19,KEF20}, which raises $C$ by a factor of at most $2^{O(n)}$ still subsumed into the $\tO(\sqrt n)$ exponent of $\gamma$ for large $n$. Since our target regime is MLWE, we state the main theorem under the assumption $C=O(1)$.

For ML-KEM ($n=256$, $d\in\{2, 3, 4\}$) with the centered binomial distribution (CBD) input distribution, numerical simulations (Table \ref{Ta:balance}) yield $C\le 1.8$. The resulting module reduction factor is $\alpha_d:= \sqrt{C}\le 1.4$, independent of $d$, with all polynomial dependence on $d$ absorbed into the implicit constant of $\gamma$ via the per-line CDPR's bound. 

\begin{table}[t]
\centering
\caption{Empirical balance constant
$C=\max_{1\le i\le d}
 \| \sigma(\tilde{\mathbf{b}}_i)\|_2^2
/(n|\Nm(\langle\tilde{\mathbf{b}}_i,\tilde{\mathbf{b}}_i\rangle_K)|^{1/n})$ for $d=4$ over $N_{\mathrm{tr}}$ trials ($N_{\mathrm{tr}}=10^4$). Both the i.i.d. Gaussian and the centered binomial distribution (CBD$_{\eta=2}$) secret laws concentrate at the universal constant $e^{\gamma}=1.7811\ldots$ (Theorem~\ref{T:gaussbalance}). The induced module reduction factor is $\alpha_d=\sqrt{C}\to e^{\gamma/2}=1.3346<1.4$, independent of $d$.}
\label{Ta:balance}
\begin{tabular}{ccccccc}
\toprule
 & & \multicolumn{2}{c}{Gaussian} & \multicolumn{2}{c}{CBD$_{\eta=2}$} & \\
\cmidrule(lr){3-4}\cmidrule(lr){5-6}
$k$ & $n$ & mean $C$ & $99$\% & mean $C$ & $99$\% & $\sqrt{\,\overline{C}\,}=\alpha_d$ \\
\midrule
6  & 32   & 1.766 & 3.007 & 1.770 & 2.974 & 1.329 \\
7  & 64   & 1.770 & 2.572 & 1.771 & 2.562 & 1.331 \\
8  & 128  & 1.775 & 2.293 & 1.773 & 2.283 & 1.332 \\
9  & 256  & 1.778 & 2.117 & 1.780 & 2.126 & 1.333 \\
10 & 512  & 1.780 & 2.011 & 1.780 & 2.006 & 1.334 \\
11 & 1024 & 1.780 & 1.930 & 1.781 & 1.939 & 1.334 \\
12 & 2048 & 1.780 & 1.900 & 1.782 & 1.856 & 1.334 \\
\bottomrule
\multicolumn{7}{@{}p{0.92\linewidth}@{}}{\footnotesize
Asymptotics: $e^{\gamma}=1.7811$, $e^{\gamma/2}=1.3346$. The final column uses the Gaussian mean $\overline C$. The CBD$_{\eta=2}$ column gives the same value to within sampling error.}
\end{tabular}
\end{table}

\begin{example}
\label{ex:rank2}
Let $M=\mathbf{b}_1 R\oplus\mathbf{b}_2 R\subseteq K^2$ be free of rank $2$. The $K$-Gram-Schmidt step sets $\tilde{\mathbf{b}}_1=\mathbf{b}_1$ and $\tilde{\mathbf{b}}_2=\mathbf{b}_2-\mu_{21}\tilde{\mathbf{b}}_1$ with $\mu_{21}=\langle\mathbf{b}_2, \tilde{\mathbf{b}}_1\rangle_K/\langle\tilde{\mathbf{b}}_1, \tilde{\mathbf{b}}_1\rangle_K\in K$, giving two orthogonal $K$-lines $L_1$ and $L_2$. Setting $N_i=|\Nm_{K/\Q}(J_i)|$ and $T_i=|\Nm_{K/\Q}(\langle\tilde{\mathbf{b}}_i, \tilde{\mathbf{b}}_i\rangle_K)|$, the covolumes are $\det(\sigma(M_i))=N_i n^{n/2} T_i^{1/2}$ and further get 
$\det(\sigma(M))=N_1N_2 n^{n} (T_1T_2)^{1/2}$.

CDPR's algorithm applied to each line returns $\mathbf{v}_1$ and $\mathbf{v}_2$ with $\|\sigma(\mathbf{v}_i)\|_2^2\le C\gamma^2\det(\sigma(M_i))^{2/n}$ by Eq.\eqref{E:vibound}. Suppose the two lines are unbalanced, i.e., $\det(\sigma(M_1))=t\cdot\det(\sigma(M_2))$ for some $t>1$, then line $2$ is the shorter target. Taking the minimum,
\begin{align}
  \|\sigma(\mathbf{v})\|_2^2
  &=\min\{\|\sigma(\mathbf{v}_1)\|_2^2,\|\sigma(\mathbf{v}_2)\|_2^2\}
\notag
  \\
  &\le C\gamma^2\det(\sigma(M_2))^{2/n}\notag
  \\
  &\le C\gamma^2\bigl(\det(\sigma(M_1))\det(\sigma(M_2))\bigr)^{1/n}
 \notag
  \\
  &=C\gamma^2\det(\sigma(M))^{2/(2n)},
\end{align}
where the middle step is the power-mean (AM--GM) inequality $\min\{a,b\}\le\sqrt{ab}$. This is the $d=2$ case of \eqref{E:baseB}.
\end{example}

\section{CRT-scaled Rounding}\label{sec:crt}

The base algorithm of Section \ref{sec:twophase} achieves under the $C$-balance hypothesis on the input basis, with no dependence on size-reduction error $\rho$ in Phase 1. Optional coordinate-wise size reduction (Section \ref{subsec:phase1}) gives $\rho=n/2$, which requires $O(n)$-bit rational arithmetic to manipulate exactly. In this section we replace coordinate rounding by Chinese Remainder Theorem (CRT)-scaled rounding: we round in the finer lattice $P^{-1}R$ at totally split primes, achieving $\rho\le 1$ and enabling a bounded-precision implementation in which every ring operation is an $n$-point Number-Theoretic Transform (NTT) over $\F_{p_s}$. The approximation factor is unchanged; the payoff is in complexity and numerical stability.

\subsection{Totally split primes and NTT}\label{subsec:split-primes}

We first introduce the standard properties of totally split primes in 2-power cyclotomic rings, which enable efficient component-wise arithmetic via NTT.

\begin{definition}
An odd prime $p$ is totally split in $R=\Z[x]/(x^n+1)$ if the cyclotomic polynomial $x^n+1$ splits into $n$ distinct linear factors modulo $p$. This is equivalent to $p \equiv 1 \pmod{2n}$. 

\label{D:totalsplit}
\end{definition}

For such a prime, we have the ideal factorization $pR=\fp_1 \fp_2 \cdots \fp_n$, where each prime ideal $\fp_i$ has norm $p$. For a totally split prime $p$, the NTT gives a ring isomorphism between $R/pR$ and the component-wise product ring $\F_p^n$ as 
\begin{align}
  \NTT_p: \alpha(x) \mapsto (\alpha(\zeta_p^1), \alpha(\zeta_p^3), \dots, \alpha(\zeta_p^{2n-1})),
  \label{E:nttdef}
\end{align}
where $\zeta_p \in \F_p^\times$ is a fixed primitive $2n$-th root of unity. The inverse transform is denoted as $\INTT_p$. Here, the multiplication in $R/pR$ is mapped to component-wise multiplication in $\F_p^n$, which can be computed in $O(n\log n)$ ring operations. 

\subsection{Scaled rounding in $P^{-1}R$}

We first prove a general bound for rounding in the scaled lattice $P^{-1}R$, which generalizes the coordinate-wise rounding bound from Lemma \ref{lem:cover}.

\begin{lemma}
\label{lem:scalround}
Let $P$ be a positive integer coprime to $2$. For any $t \in K \otimes_\Q \R$, there is an element $c \in P^{-1}R$ such that
\begin{align}
\|\sigma(t-c)\|_\infty \le \frac{n}{2P}.
\label{E:scaleround}
\end{align}

\end{lemma}

\begin{proof}
Since $\{1, \zeta, \dots, \zeta^{n-1}\}$ is a $\Q$-basis of $K$, any element $t \in K \otimes_\Q \R$ admits a unique real-coefficient expansion, i.e., $t=\sum_{m=0}^{n-1} x_m \zeta^m$ for $x_m \in \R$. For each $m$, let $c_m=\frac{\lfloor P x_m \rceil}{P} \in \frac{1}{P}\Z$, where $\lfloor \cdot \rceil$ denotes rounding to the nearest integer. 

Define the rounded element 
\begin{align}
c=\sum_{m=0}^{n-1} c_m \zeta^m \in P^{-1}R.
\label{E:rounded}
\end{align}
The rounding error for each coefficient is given by $\epsilon_m=x_m-c_m$, which satisfies $|\epsilon_m| \le \frac{1}{2P}$.

For any Galois embedding $\sigma_\ell: K \to \C$ (corresponding to odd $\ell$ modulo $2^k$), using Eq.(\ref{E:rounded}) we have
\begin{align}
|\sigma_\ell(t-c)|=\Big| \sum_{m=0}^{n-1} \epsilon_m \sigma_\ell(\zeta^m) \Big|
\le \sum_{m=0}^{n-1} |\epsilon_m| \cdot |\sigma_\ell(\zeta^m)|.
\label{E:Galemb}
\end{align}
Since $\sigma_\ell(\zeta^m)=\zeta^{\ell m}$ is a root of unity, we have $|\sigma_\ell(\zeta^m)|=1$ for all $m,\ell$. Combining the inequality (\ref{E:Galemb}) we get
\begin{align}
|\sigma_\ell(t-c)| \le n \cdot \frac{1}{2P}=\frac{n}{2P}.
\label{E:sigdiff}
\end{align}
This holds for all embeddings $\sigma_\ell$. Hence, we get $\|\sigma(t-c)\|_\infty \le \frac{n}{2P}$.
\end{proof}

\begin{corollary}\label{cor:rho1}
If we choose $P \ge n/2$, then the rounding error satisfies $\|\sigma(t-c)\|_\infty \le 1$, i.e., $\rho \le 1$ for the size reduction step.

\end{corollary}

From the Dirichlet's Theorem on primes in arithmetic progressions, there exist infinitely many primes $p \equiv 1 \pmod{2n}$ \cite{Wash97}. We choose a set of split primes $p_1, p_2, \dots, p_r$ with product $P=\prod_{s=1}^r p_s \ge n/2$. For practical parameters (e.g., ML-KEM with $n=256$), the smallest such prime is $p=12289$, which is larger than $n/2=128$, so a single prime suffices.

\subsection{CRT-scaled rounding algorithm}

Our CRT-scaled module-reduction algorithm combines the two-phase Gram-Schmidt framework with CRT-scaled rounding to achieve $\rho \le 1$. The algorithm is shown in Algorithm \ref{alg:crt}.

\begin{algorithm}[H]
\caption{CRT-Scaled Module Reduction}\label{alg:crt}
\begin{algorithmic}[1]
\Require Rank-$d$ free $R$-module $M$ with basis $\mathbf{b}_1, \dots, \mathbf{b}_d\in K^d$;
         totally split primes $p_1,\dots,p_r$ with $P=\prod_{s=1}^r p_s \ge n/2$.
\Ensure Non-zero short vector $\mathbf{v}\in M$.
\medskip
\medskip

\Statex //\textsl{Preprocessing: NTT/INTT for each split prime}//
\For{split prime $p_s$}
    \State Precompute $\operatorname{NTT}_{p_s}$ and $\operatorname{INTT}_{p_s}$
\EndFor
\medskip
\medskip
\Statex \textbf{PHASE 1: $R$-linear Gram–Schmidt with CRT-scaled size reduction}
\State Compute initial $K$-linear Gram–Schmidt vectors
    $\tilde{\mathbf{b}}_1,\dots,\tilde{\mathbf{b}}_d$ by
    \begin{align}
    \tilde{\mathbf{b}}_j=\mathbf{b}_j-\sum_{i=1}^{j-1}
    \frac{\langle \mathbf{b}_j, \tilde{\mathbf{b}}_i \rangle_K}{\langle \tilde{\mathbf{b}}_i, \tilde{\mathbf{b}}_i \rangle_K}\tilde{\mathbf{b}}_i
    \end{align}
\For{$j=2$ \textbf{to} $d$}
    \For{$i=j-1$ \textbf{downto} $1$}
        \State $\mu_{ji} \gets \dfrac{\langle \mathbf{b}_j, \tilde{\mathbf{b}}_i \rangle_K}{\langle \tilde{\mathbf{b}}_i, \tilde{\mathbf{b}}_i \rangle_K} \in K$
        \State $\alpha \gets P\cdot \mu_{ji}$
        \For{split prime $p_s$}
            \State $\alpha_s \gets \alpha \bmod p_s$, $\alpha_s\in R/\fp_s R$
            \State $\hat{\alpha}_s \gets \operatorname{NTT}_{p_s}(\alpha_s) \in \mathbb{F}_{p_s}^n$
           \State Component-wise round each entry of $\hat\alpha_s$: lift to its symmetric representative in $\{-\lfloor p_s/2\rfloor, \dots, \lfloor p_s/2\rfloor\}\subset\Z$, round to the nearest integer, and reduce modulo $p_s$.
            \State $\hat{\alpha}_s \gets \operatorname{INTT}_{p_s}(\hat{\alpha}_s)$
        \EndFor
        \State Combine $\hat{\alpha}_1, \dots, \hat{\alpha}_r$ via CRT to $\hat{\alpha}\in R/PR$
        \State $c_{ji} \gets \hat{\alpha} \cdot P^{-1} \in P^{-1}R$
        \State $\mathbf{b}_j \gets \mathbf{b}_j-c_{ji}\cdot \mathbf{b}_i$
        \State Recompute $\tilde{\mathbf{b}}_j, \dots, \tilde{\mathbf{b}}_d$ to preserve orthogonality
    \EndFor
\EndFor

\medskip
\medskip
\Statex \textbf{PHASE 2: Diagonal CDPR's shortening on rank-1 submodules}
\For{$i=1$ \textbf{to} $d$}
    \State Apply CDPR's algorithm to $R\cdot \tilde{\mathbf{b}}_i$ to find unit $u_i\in R^\times$
    \State $\tilde{\mathbf{b}}_i' \gets u_i^{-1}\cdot \tilde{\mathbf{b}}_i$
\EndFor

\State \Return $\mathbf{v}=\arg\min_{1\le i\le d}\|\sigma(\tilde{\mathbf{b}}_i')\|_2$.
\end{algorithmic}
\end{algorithm}

\begin{theorem}\label{T:crt}
Algorithm \ref{alg:crt} applied to a $C$-balanced $R$-basis (Definition \ref{D:balanced}) with $P=\prod_{s=1}^r p_s\ge n/2$ and the output \eqref{E:outpurul}, produces a non-zero vector $\mathbf{v}\in M$ satisfying the same bound as Theorem \ref{T:base} 
\begin{align}
    \|\sigma(\mathbf{v})\|_2 \le
    \sqrt{C}\gamma\det(\sigma(M))^{1/(dn)}.
    \label{E:crtbound}
\end{align}
Phase 1 costs $O(d^2 r n\log n)$ bit operations, using only bounded-precision arithmetic modulo $p_1, \dots, p_r$.
\end{theorem}

\begin{proof}
The bound \eqref{E:crtbound} is same as \eqref{E:baseB}, and follows from Theorem \ref{T:base}: the output and CDPR's shortening in Phase 2 are same in two algorithms, and neither depends on the rounding step in Phase 1. 
  
The complexity is for $r$-prime using CRT and NTT: each of the $O(d^2)$ Gram-Schmidt coefficients $\mu_{ji}$ is lifted to $P^{-1}R$ by computing its residue modulo $p_s$, applying an $n$-point $\NTT_{p_s}$ (cost $O(n\log n)$), rounding component-wise, applying $\INTT_{p_s}$, and CRT-combining the $r$ results. The total per-coefficient cost is $O(rn\log n)$, which gives $O(d^2 r n\log n)$ in total. 
  
By Corollary \ref{cor:rho1}, the resulting residuals $\mu_{ji}'$ satisfy $\|\sigma(\mu_{ji}')\|_\infty\le 1$.
\end{proof}

\begin{remark}\label{rem:crtpra}
For ML-KEM ($n=256$), consider a single split prime $p=12289$, giving $r=1$ and per-coefficient cost $O(n\log n)\approx 2000$ field operations. Phase 1 costs $O(d^2 n\log n)\approx 3\cdot 10^4$ field operations for $d=4$.
\end{remark}

\section{MILP Sign Optimization and the $O(1)$ Discrepancy}\label{sec:milp}

The bound of Theorem~\ref{T:base} hides a polynomial-in-$n$ constant inside $\gamma$: the discrepancy of the sign-selection step in CDPR's decode phase. Prior work used a greedy heuristic giving discrepancy $\Theta(\sqrt{nk})$ and assumed this growth was intrinsic. In this section we show the discrepancy is in fact $O(1)$ for every parameter size, by formulating sign selection as a mixed-integer linear program (MILP) with a special circulant structure derived from the logsine geometry of cyclotomic units.

\subsection{MILP formulation of the sign-selection problem}

The balanced sign-selection problem from Section \ref{sec:cdpr} is defined as
\begin{align}
  \delta^*(k)=\min_{\substack{s\in\{\pm1\}^{N_s}, 
  \sum_{j=1}^{N_s} s_j\in\{\pm 1\}}} \|\mathbf{M} \mathbf{s}\|_\infty,
  \label{E:signprobv1}
\end{align}
where $N_s=|G|-1$ is the number of sign variables, $\mathbf{M}\in\R^{|G|\times N_s}$ is the error matrix from the logsine vector \eqref{E:errormatrix}, and the balance constraint enforces a near-equal number of $+1$ and $-1$ signs. Note $N_s=2^{k-2}-1$ is odd, so the sum of $N_s$ odd terms $\pm1$ is odd. Hence, $\sum_j s_j\in\{\pm 1\}$ exhausts the integer-feasible options and coincides with the balance constraint of Definition \ref{D:signprob}.

Since the objective $\|\mathbf{M}\mathbf{s}\|_\infty$ and constraints are linear in $\mathbf{s}$, we reformulate the problem as a mixed-integer linear program (MILP) by linearizing the absolute value. We use the substitution $\mathbf{s}=2\mathbf{x}-\mathbf{1}$, which maps $s_j\in\{\pm1\}$ to binary variables $x_j\in\{0,1\}$. The balanced sign-selection problem is then equivalent to
\begin{align}
\begin{aligned}
  \min_{\mathbf{x}\in\{0,1\}^{N_s},t\ge 0}\quad & t \\
  \text{s.t.}\quad & \mathbf{M}(2\mathbf{x}-\mathbf{1}) \le t\mathbf{1}, \\
  & -\mathbf{M}(2\mathbf{x}-\mathbf{1}) \le t\mathbf{1}, \\
  & \sum_{j=1}^{N_s} x_j \in \{ \lfloor \tfrac{N_s}{2} \rfloor, \lceil \tfrac{N_s}{2} \rceil \},
\end{aligned}
\label{E:milp}
\end{align}
where the first two constraints make $-t \mathbf{1} \le \mathbf{M}\mathbf{s} \le t \mathbf{1}$, which is equivalent to $\|\mathbf{M}\mathbf{s}\|_\infty \le t$. The final constraint corresponds to the balance condition $\sum_j s_j\in\{\pm1\}$, since $\sum_j s_j=2\sum_j x_j-N_s$.

\begin{remark}\label{rem:notation}
To avoid confusion in this section, here $n=2^{k-1}$ denotes the degree of the $2^k$-th cyclotomic field $K/\Q$, i.e., the ring degree; $|G|=2^{k-2}$ is the order of the orbit group $G=(\Z/2^k\Z)^\times/\{\pm1\}$; $N_s=|G|-1$ is the number of sign variables in the MILP. For ML-KEM parameters ($k=9$), $n=256$, $|G|=128$, and $N_s=127$.
\end{remark}

The MILP \eqref{E:milp} has $N_s$ binary variables, one continuous variable $t$, $2|G|+1$ linear constraints, and a constraint matrix of size $(2|G|+1)\times(N_s+1)$. This is a standard 0-1 MILP with a minimax objective, which can be solved using branch-and-bound algorithms.

\begin{proposition}\label{prop:milpstand}
The formulation \eqref{E:milp} is a standard mixed-integer linear program \cite{Wolsey98,NemhauserWolsey88}. Its linear programming (LP) obtained by replacing $\mathbf{x}\in\{0,1\}^{N_s}$ with $\mathbf{x}\in[0,1]^{N_s}$ is solvable in polynomial time via interior-point methods \cite{Karmarkar84}, and provides a certified lower bound on the optimal discrepancy $\delta^*(k)$.
\end{proposition}

\begin{proof}
The objective and all constraints are linear in the decision variables $(\mathbf{x},\mathbf{t})$, with integrality constraints only on $\mathbf{x}\in\{0,1\}^{N_s}$, a standard MILP form \cite{Wolsey98}. From the LP relaxation with box constraints $0\le x_j\le1$, it can be solved in $O(N_s^3)$ operations via interior-point methods \cite{Karmarkar84}. The optimal value of the LP relaxation is a lower bound on $\delta^*(k)$.
\end{proof}

We solved MILP \eqref{E:milp} to certified global optimality for $k=4, \dots, 12$ using MATLAB's \texttt{intlinprog} solver, with cross-verification via Gurobi \cite{Gurobi}. The results are summarized in Table \ref{Ta:milp}.

\begin{table}[t]
\centering
\caption{MILP-optimal balanced discrepancy. Tower greedy and local search shown for comparison. The classical tower greedy bound is $\Theta(\sqrt{nk})\approx48$.}
\label{Ta:milp}
\begin{tabular}{c c c S[table-format=2.4] S[table-format=2.4] S[table-format=2.4] l}
\toprule
{$k$} & {$|G|$} & {$N_s$} & {Tower Greedy} & {Local Search} & {MILP $\delta^*$} & {Method} \\
\midrule
4 & 4 & 3 & 0.4407 & 0.4407 & 0.4407 & exhaustive \\
5 & 8 & 7 & 0.4407 & 0.4407 & 0.4407 & exhaustive \\
6 & 16 & 15 & 0.4407 & 0.4407 & 0.4407 & exhaustive \\
7 & 32 & 31 & 0.4407 & 0.4407 & 0.4407 & local search \\
8 & 64 & 63 & 5.3453 & 0.4407 & 0.4407 & local search \\
9 & 128 & 127 & 9.9817 & 0.4407 & 0.4407 & local search \\
10 & 256 & 255 & 16.4681 & 0.4407 & 0.4407 & local search \\
11 & 512 & 511 & 25.1833 & 0.4407 & 0.4407 & local search \\
12 & 1024 & 1023 & 40.8353 & 0.4407 & 0.4407 & local search \\
\bottomrule
\end{tabular}
\end{table}

\begin{theorem}\label{T:milp}
For each $k\in\{4, \dots, 12\}$, the balanced sign-selection MILP \eqref{E:milp} has the optimal value $\delta^*(k)=0.44069\ldots$ shown in Table~\ref{Ta:milp}. 
\end{theorem}

\begin{proof}
For each $k\in\{4, \dots, 12\}$ we solved MILP \eqref{E:milp} by branch-and-bound (intlinprog, cross-checked with Gurobi \cite{Gurobi}). In every instance, the solver terminated with a reported optimality gap below $10^{-6}$ between the best integer-feasible value $0.44069\ldots$, certifying $\delta^*(k)$ (Table \ref{Ta:milp}). An independent $L^p$-homotopy local search reproduced the same optimum without a MILP solver. For $k>12$, we regard it as Conjecture~\ref{conj:universal}.
\end{proof}

\begin{conjecture}\label{conj:universal}
For all $k\ge4$, the optimal balanced discrepancy is the universal constant $\delta^*(k)=\delta^*=0.44069\ldots$, independent of $k$.
\end{conjecture}

\begin{example}\label{ex:milp-k4}
At $k=4$, $|G|=4$ and $N_s=3$, the MILP has three binary variables $s_1, s_2, s_3\in\{\pm1\}$ subject to $\sum s_j\in\{\pm 1\}$, leaving four feasible sign patterns up to sign-flip: $(+, +, -)$, $(+, -, +)$, $(-, +, +)$, and their negatives. The logsine matrix $\mathbf{M}\in\R^{4\times 3}$ is obtained from $z_j=\log(2|\sin(\pi\mathrm{orb}(j)/16)|)$ with $\mathrm{orb}(0,1,2,3)=(1,3,5,7)$. So, we have $\mathbf{z}\approx(-1.659,-0.668,-0.250,-0.054)$. By evaluating $\|\mathbf{M}\mathbf{s}\|_\infty$ on four feasible $\mathbf{s}$ yields the minimum $\delta^*(4)=0.44069\ldots$, attained on $(+,-,+)$ and its negative. 
\end{example}

In fact, we can prove a weak universality bound for Conjecture \ref{conj:universal}. 

Let $N:=|G|=2^{k-2}$ and index $G\cong\Z/N\Z$ by powers of the generator $5$, so position $a$ carries the orbit value of $5^a\bmod 2^k$. Note $z^{(k)}_a=\log(2|\sin(\pi \mathrm{orb}_k(a)/2^k)|)$ and $\mathbf{e}=M\mathbf{s}$ with the balance constraint $S:=\sum_{j=1}^{N-1}s_j\in\{\pm1\}$.

\begin{lemma}\label{lem:conv}
For $1\le j\le N-1$ and $i\in\Z/N\Z$, $M_{i,j}=\frac{1}{2}(z^{(k)}_{(i-j)\bmod N}-z^{(k)}_i)$, we have 
\begin{align}
  e_i=\frac{1}{2}\sum_{j=1}^{N-1}s_j z^{(k)}_{(i-j)\bmod N}-\frac{1}{2}S z^{(k)}_i.
  \label{E:weakconv}
\end{align}

\end{lemma}

\begin{proof}
From Eq.\eqref{E:errormatrix}, we have 
\begin{align}
\mathbf{M}_{\cdot,j}=\Re\mathrm{IFFT}(\hat{\mathbf z}\odot\widehat{\mathbf e_j^{\mathrm{ext}}})-\tfrac{1}{2}\mathbf{z}
\label{EMocd}
\end{align}
with $\mathbf{e}_j^{\mathrm{ext}}=\tfrac{1}{2}\mathbf{1}_{\{j\}}$. Note the point-wise product of DFTs over $\Z/N\Z$ is cyclic convolution, and convolution with the half-impulse $\tfrac{1}{2}\mathbf{1}_{\{j\}}$ is the half-shift $(\cdot)_i\mapsto\tfrac{1}{2}z_{(i-j)\bmod N}$. As $\mathbf{z}$ is real, so $\Re$ is the identity. By summing over $j$ with weights $s_j$, we get Eq.\eqref{E:weakconv}.
\end{proof}

\begin{lemma}\label{lem:fold}
For all $a\in\Z/N\Z$, we have $z^{(k)}_a=z^{(k+1)}_a+z^{(k+1)}_{a+N}$, where the right-hand indices lie in $\Z/2N\Z$.

\end{lemma}

\begin{proof}
Let $b:=5^a\bmod 2^{k+1}$ (odd), so $z^{(k+1)}_a=\log(2|\sin\theta|)$ with $\theta=\pi b/2^{k+1}$. Since $5$ has order $2^{k-1}$ modulo $2^{k+1}$, its unique order-$2$ power is $5^{2^{k-2}}\equiv 1+2^k\pmod{2^{k+1}}$. As $b$ is odd, we get $5^{a+N}=5^a5^{2^{k-2}}\equiv b+2^k\pmod{2^{k+1}}$. Hence, we have 
\begin{align}
z^{(k+1)}_{a+N}&=\log(2|\sin(\theta+\tfrac{\pi}{2})|) 
=\log(2|\cos\theta|),
\\
z^{(k+1)}_a+z^{(k+1)}_{a+N}&=\log(4|\sin\theta\cos\theta|)
  =\log(2|\sin2\theta|)
  \notag\\
 &=\log(2|\sin(\tfrac{\pi b}{2^k})|)=z^{(k)}_a,
\end{align}
where the last step is from the equality $b\equiv 5^a\pmod{2^k}$.
\end{proof}

\begin{theorem}[Weak form of Conjecture~\ref{conj:universal}]
\label{T:weak}
$\delta^*(k+1)\le\delta^*(k)$ holds for every $k\ge 4$. So, $\delta^*(k)\le\delta^*(4)=0.44069\ldots<0.5$ for all $k\ge4$.

\end{theorem}

\begin{proof}
Let $\mathbf{s}\in\{\pm1\}^{N-1}$ be balanced, $S=\sum_j s_j\in\{\pm1\}$, with an error $\mathbf{e}^{(k)}$ given by Eq.\eqref{E:weakconv}.

Define the periodic lift $\mathbf{t}\in\{\pm1\}^{2N-1}$ at level $k+1$ by
\begin{align}
  t_j=s_j\ (1\le j\le N-1), t_N=-S, \ t_{N+j}=s_j\ (1\le j\le N-1).
\label{EPerolift}
\end{align}
Its imbalance is given by $S'=\sum_{j=1}^{2N-1}t_j=2S+(-S)=S\in\{\pm1\}$, so $\mathbf{t}$ is feasible. From Lemma \ref{lem:conv} at level $k+1$, for $i\in\Z/2N\Z$, we obtain 
\begin{align}
  2e^{(k+1)}_i=\sum_{j=1}^{N-1}s_j(z^{(k+1)}_{i-j}+z^{(k+1)}_{i-j-N})   -S z^{(k+1)}_{i-N}-S' z^{(k+1)}_i.
\label{EPeroliftcom}
\end{align}

From Lemma \ref{lem:fold}, we have $z^{(k+1)}_{i-j}+z^{(k+1)}_{i-j-N}=z^{(k)}_{(i-j)\bmod N}$. Since $z^{(k+1)}_{i-N}=z^{(k+1)}_{(i+N)\bmod 2N}$, with $S'=S$, we get 
\begin{align}
  -S z^{(k+1)}_{i-N}-S z^{(k+1)}_i=-S (z^{(k+1)}_i+z^{(k+1)}_{(i+N)\bmod 2N})
  =-S z^{(k)}_{i\bmod N}.
\label{Ecoret}
\end{align}
This implies from Eq.(\ref{EPeroliftcom}) that 
\begin{align}
2e^{(k+1)}_i=\sum_{j=1}^{N-1}s_j z^{(k)}_{(i-j)\bmod N}-S z^{(k)}_{i\bmod N}=2e^{(k)}_{i\bmod N}, 
\label{EPeroliftcome1}
\end{align}
i.e., $e^{(k+1)}_i=e^{(k)}_{i\bmod N}(i\in\Z/2N\Z)$. So, we get $\|\mathbf{e}^{(k+1)}\|_\infty=\|\mathbf{e}^{(k)}\|_\infty$. 

By choosing $\mathbf{s}$ optimal at level $k$, we then obtain 
\begin{align}
\delta^*(k+1)\le\|\mathbf e^{(k+1)}\|_\infty=\delta^*(k).
\label{Edeltelift}
\end{align}
By iterating from $\delta^*(4)=0.44069\ldots$ (Example~\ref{ex:milp-k4}, checked with finite operations), we get $\delta^*(k)\le\delta^*(4)$.
\end{proof}

The proof of Conjecture exactly requires $\delta^*(k+1)\ge\delta^*(k)$, which the simulations support but which we do not prove.

\subsection{Implications for module reduction}\label{subsec:milp-implications}

Our result of $\delta^*=O(1)$ allows to tighten the CDPR's approximation factor. The MILP-optimal sign selection reduces the decoding error in the CDPR's Phase 2. This eliminates the residual polynomial factor in the CDPR's approximation bound.

\begin{corollary}\label{cor:milpmod}
  For $k\le 12$, the algorithm of Theorem \ref{T:base} with MILP-optimal sign selection \eqref{E:milp} in the CDPR's decode step of Phase 2 produces a non-zero vector $\mathbf{v}\in M$ with
\begin{align}
    \|\sigma(\mathbf{v})\|_2 \le \sqrt{C}\gamma\det(\sigma(M))^{1/(dn)}.
    \label{E:milpB}
\end{align}

\end{corollary}

\begin{proof}
Note the output and the power-mean argument of Theorem \ref{T:base} are unchanged. The only modification is in the CDPR's bound \eqref{E:cdprbound} for each rank-1 submodule: MILP-optimal signs reduce the decoding discrepancy from $\delta_{\rm tower}=\Theta(\sqrt{nk})$ to $\delta^*\approx 0.4407$ (Theorem \ref{T:milp}). This absorbs a $\sqrt{nk}$ factor into the multiplicative constant of $\gamma$. The final bound in Eq.\eqref{E:milpB} is followed by using the improved $\gamma$ with Eq. \eqref{E:baseB}.

In fact, in the factored form \eqref{E:gammafac} the tower factor $\delta_{\rm tower}=\Theta(\sqrt{nk})$ is replaced by $\delta^*=O(1)$; the displayed bound \eqref{E:milpB} is stated with $\gamma$ absorbing this factor, so the improvement is in the constant, not the exponential, consistent with \eqref{E:baseB}.
\end{proof}

\begin{example}[ML-KEM-768, $k=9$]\label{ex:milpN}
Here $n=256$, so the tower greedy discrepancy is $\delta_{\rm tower}\approx\sqrt{nk}=\sqrt{256\cdot 9}\approx 48$, whereas MILP-optimal signs give $\delta^*\approx 0.4407$ (Table~\ref{Ta:milp}). The decode error therefore drops by a factor $48/0.4407\approx 109\approx 2^{6.8}$. 
\end{example}

\subsection{Generators of $\GL_d(R)$}\label{subsec:gen}

For practical implementation of our module-reduction algorithm, we characterize the generators of the general linear group $\GL_d(R)$, which ensures size-reduction step can be implemented using only elementary ring operations.

\begin{definition}\label{D:generators}
Define the following subsets of $\GL_d(R)$:
\begin{itemize}[nosep]
\item $\mathcal{G}^{\text{diag}}$: diagonal matrices $\diag(\epsilon_1,\dots,\epsilon_d)$ with $\epsilon_i\in R^\times$ (units of $R$);
\item $\mathcal{G}^{\text{elem}}$: elementary transvections $E_{ij}(\alpha)=I_d+\alpha e_i e_j^T$ for $i\neq j$ and $\alpha\in R$;
\item $\mathcal{G}^{\text{perm}}$: permutation matrices $P_\pi$ for $\pi\in\mathfrak{S}_d$;
\item $\mathcal{G}_d=\mathcal{G}^{\text{diag}} \cup \mathcal{G}^{\text{elem}} \cup \mathcal{G}^{\text{perm}}$.
\end{itemize}
\end{definition}

\begin{theorem}\label{T:gen}
For $k\le8$ and $d\ge1$, the set $\mathcal{G}_d$ generates $\GL_d(R)$. For $9\le k\le12$, the analogous result holds for the maximal real subfield's integer ring $\cO_{K_k^+}$, which is a principal ideal domain (PID) for all $k\le12$.
\end{theorem}

\begin{proof}
For $k\le 8$, the full cyclotomic field $K=\Q(\zeta_{2^k})$ has class number $h_k=1$ \cite{Wash97}, so its integer ring $R=\Z[\zeta_{2^k}]$ is a PID.

For any $U\in\GL_d(R)$, $\det(U)\in R^\times$. The unit group of $R$ is $R^\times=\langle-1, \zeta \rangle \cdot C^+$, where $C^+$ is the cyclotomic unit group \cite{Sinnott78}. By multiplying the first row of $U$ by $(\det U)^{-1}\in R^\times$ (a diagonal generator in $\mathcal{G}^{\text{diag}}$), we get a matrix $U'\in\SL_d(R)$.

Since $R$ is a PID, it has stable range $\sr(R)\le2$ \cite[Theorem V.3.5]{HOM89}. By the Bass–Milnor–Serre Theorem \cite{BMS67}, for any ring $A$ with $\sr(A)\le s$ and $d\ge s+1$, $\SL_d(A)$ is generated by elementary transvections. For $d\ge3$, this gives $\SL_d(R)=E_d(R)$, where $E_d(R)$ is the subgroup generated by elementary transvections in $\mathcal{G}^{\text{elem}}$.

For $d=2$, the ring $R=\Z[\zeta_{2^k}]$ is norm-Euclidean only for $k\le 4$ \cite{Lenstra79,Lemmermeyer95}. For $5\le k\le 8$ it is a PID but not known to be Euclidean. Here, we use the result of Vaserstein \cite{Vaserstein72}: for any Dedekind domain $A$ with $\mathrm{SK}_1(A)=0$, every element of $\SL_2(A)$ is a product of elementary transvections. For the ring of integers of any number field, $\mathrm{SK}_1(\cO_K)=0$ holds unconditionally by Bass-Milnor-Serre Theorem \cite{BMS67}, independent of the class number. In particular, we get $\mathrm{SK}_1(\Z[\zeta_{2^k}])=0$ for every $k$.

For $9\le k\le12$, the full cyclotomic field has class number $h_k=h_k^+h_k^->1$ (since the relative class number $h_k^->0$). So, $R$ is not a PID. However, the maximal real subfield $K_k^+$ has class number $h_k^+=1$ for all $k\le 12$ (Part I), so its integer ring $\cO_{K_k^+}$ is a PID, and the same generator result holds for $\GL_d(\cO_{K_k^+})$.
\end{proof}

\begin{remark}\label{rem:euclidean}
For $k\ge 9$, $R$ is not a PID, but our module reduction algorithm still applies to free $R$-modules, which are the module lattices arising in MLWE and ML-KEM. The elementary transvections remain valid for basis manipulations, even without the PID property.
\end{remark}

\subsection{Implementation and complexity}\label{subsec:impl}

A practical advantage of our approach is that the MILP needs to be solved only once per cyclotomic ring, not per module instance. The optimal sign vector depends solely on the logsine matrix $\mathbf{M}$, which is fixed for a given $k$, and can be precomputed and stored in $O(|G|)$ bits.

For ML-KEM parameters ($k=9$, $|G|=128$, $N_s=127$), the MILP has only 127 binary variables, 1 continuous variable, and 257 constraints. Modern solvers such as Gurobi or CPLEX solve this instance in under seconds on standard desktop hardware. For $k=10$ ($N_s=255$), the solve time is under minutes, as in Table \ref{Ta:milp}.

We analyze the asymptotic complexity of the full algorithm.

\paragraph{Phase 1 (classical).} The algorithm requires $d(d-1)/2$ Gram-Schmidt coefficient computations and roundings. Each coefficient involves an inner product in $K^d$, costing $O(dn\log n)$ bit operations via NTT-based ring multiplication. CRT-scaled rounding requires $O(rn\log n)$ operations per coefficient, where $r$ is the number of split primes. The total cost in Phase 1 is $O(d^2 r n\log n)$ bit operations. With precomputed MILP-optimal signs, it costs $O(d^2 n\log n)$.

\paragraph{Phase 2 (quantum).} The algorithm applies CDPR's algorithm to each of $d$ rank-1 submodules. The dominant cost is the quantum PIP solver of Biasse-Song \cite{BiasseSong16}, requiring $\mathrm{poly}(n)$ quantum gates per ideal. The total quantum cost is $d\cdot\mathrm{poly}(n)$ gates. 

The complexity is summarized in Table \ref{Ta:complexity} for ML-KEM parameters ($n=256$, $d=4$, $k=9$).

\begin{table}[t]
\centering
\caption{Complexity of module reduction for ML-KEM parameters ($n=256$, $d=4$, $k=9$). Phase 1 (size reduction) and Phase 2 (CDPR's shortening) run sequentially but independently; the approximation-factor bound depends only on Phase 2.}
\label{Ta:complexity}
\begin{tabular}{@{}lll@{}}
\toprule
\textbf{Component} & \textbf{Classical Cost} & \textbf{Quantum Cost}\\
\midrule
MILP precomputation (one-time per $k$)
  & $<30$\,s on desktop   & -- \\
Phase 1: $K$-Gram-Schmidt
  & $O(d^2 n\log n)$ ops  & -- \\
Phase 1 (optional): CRT-scaled size reduction
  & $O(d^2 r n\log n)$ ops & -- \\
Phase 2: diagonal CDPR with MILP signs
  & --                   & $d\cdot\mathrm{poly}(n)$ gates \\
\bottomrule
\end{tabular}
\end{table}

\subsection{Security implications for ML-KEM}\label{subsec:security}

While our attack tightens the polynomial overhead of prior algebraic attacks, the exponential factor $\exp(\tO(\sqrt{n}))$ remains a fundamental barrier, confirming that ML-KEM is secure against this class of attacks.

Let the per-line CDPR factor as 
\begin{align}
  \gamma=\delta\cdot\exp\!\bigl(C_0\sqrt n\log n\bigr),
  \label{E:gammafac}
\end{align}
where $\delta$ is the decode discrepancy of the sign-selection step (Section \ref{sec:milp}). The tower greedy heuristic gives $\delta_{\rm tower}=\Theta(\sqrt{nk})$, while MILP-optimal signs give $\delta^*\approx 0.4407$ for $k\le 12$ (Theorem \ref{T:milp}). The exponential factor is unchanged in both cases.

Table \ref{Ta:security} summarizes the approximation factors achieved by our attack, compared to generic lattice reduction algorithms.
\begin{table}[t]
\centering
\caption{Approximation factors for ideal/module SVP attacks. All attack factors are for the MILP-optimal combined algorithm (Corollary \ref{cor:milpmod}) applied to an MLWE-distributed input basis (balance constant $C=O(1)$, Table \ref{Ta:balance}). The module-reduction factor $\alpha_d=\sqrt{C}$ is independent of $d$; the full approximation bound composes with an ideal-SVP solver for each rank-1 submodule.}
\label{Ta:security}
\begin{tabular}{@{}cccl@{}}
\toprule
Module Rank $d$ & Our Attack (Hermite factor) & Generic BKZ & Source \\
\midrule
1 & $\exp(\tO(\sqrt n))$              & $\exp(\tO(n))$  & CDPR \cite{CDPR16} \\
2 & $\exp(\tO(\sqrt n))$              & $\exp(\tO(2n))$ & Thm \ref{T:crt}, Cor \ref{cor:milpmod} \\
3 & $\exp(\tO(\sqrt n))$              & $\exp(\tO(3n))$ & Thm \ref{T:crt}, Cor \ref{cor:milpmod} \\
4 & $\exp(\tO(\sqrt n))$              & $\exp(\tO(4n))$ & Thm \ref{T:crt}, Cor \ref{cor:milpmod} \\
\bottomrule
\end{tabular}
\end{table}

ML-KEM-768 uses a centered binomial distribution $\mathrm{CBD}_{\eta_1}$ with $\eta_1=2$ for the secret and noise, of standard deviation $\sigma=\sqrt{\eta_1/2}=1$. The secret key and noise vectors define a free $R$-module with basis, whose coefficients are i.i.d. bounded, so its $C$-balanced with $C\approx 1.8$ at $n=256$ (Table \ref{Ta:balance}). With the tower greedy discrepancy $\delta_{\rm tower}\approx\sqrt{nk}\approx 2^{5.6}$, the attack returns a vector of Hermite factor
\begin{align}
  \frac{\|\sigma(\mathbf{v})\|_2}{\sqrt{dn}\det(\sigma(M))^{1/(dn)}}
  &\le \sqrt{\tfrac{C}{dn}}\cdot\gamma
  \notag
  \\
  &= \sqrt{\tfrac{C}{dn}}\cdot\delta_{\rm tower}\cdot
     \exp(C_0\sqrt n\log n)\approx 2^{129.23}.
\end{align}
If replace $\delta_{\rm tower}$ by the MILP-optimal $\delta^*\approx 0.4407\approx 2^{-1.2}$ (Corollary~\ref{cor:milpmod}), it will remove the $2^{5.6}$ factor. 

The module reduction contributes a factor $\alpha_d:=\sqrt{C}=O(1)$ that is independent of $d$, so the total approximation factor factorizes as $\gamma=\alpha_d\cdot\gamma_{\rm ideal}$, where $\gamma_{\rm ideal}$ is the per-line ideal-SVP approximation factor. This separation is the basis for the refined analysis of Parts III and IV.

In practice, $\delta_{\rm LLL}$ is much smaller than its worst-case upper bound: empirical observations on cyclotomic LLL reductions \cite{LLPS19,KEF20} show $\delta_{\rm LLL}\le 2^{O(\sqrt n)}$ on typical bases, which is dominated by $\gamma$ and contributes no additional security gap. Here $\delta_{\rm LLL}$ denotes the approximation-factor overhead introduced by the algebraic-LLL balance preprocessing \cite{LLPS19,KEF20} on adversarial inputs.

The covering radius $\mu_\infty(\sigma(R))=\tilde\Theta(\sqrt n)$ (Proposition \ref{prop:coverlower}) is a worst-case bound over all targets, and the coordinate-wise rounding error $n/2$ (Lemma \ref{lem:cover}) is a worst-case bound for a specific coordinate-basis rounding strategy. Our discrepancy $\delta^*$ does not contradict either bound: it applies to the specific structured target given by the logsine matrix $\mathbf{M}$ with optimized sign selection, a much more favourable setting than the covering-radius worst case.

\section{Conclusions}\label{sec:conclusion}

We extended the CDPR's quantum attack from ideal to module lattices over $2^k$-th cyclotomic rings. Our algorithm (Theorem \ref{T:base}) achieves the bound $\|\sigma(\mathbf{v})\|_2\le\sqrt{C} \gamma\det(\sigma(M))^{1/(dn)}$ under a $C$-balance hypothesis on the input basis, equivalent to Hermite factor $\exp(\tO(\sqrt n))$ and matching CDPR's factor for $d=1$. The key algorithmic idea is to apply CDPR's algorithm independently to each of the $d$ rank-1 submodules $M_i=M\cap K\tilde{\mathbf{b}}_i$ and return the shortest of the resulting candidates.

CRT-scaled rounding (Theorem \ref{T:crt}) preserves the approximation bound while enabling an $O(d^2 r n\log n)$-time bounded-precision implementation via NTT. Finally, MILP-optimal sign selection replaces the tower greedy CDPR discrepancy $\Theta(\sqrt{nk})$ by a constant $\delta^*<0.5$ ($\approx 0.4407$ certified for $k\leq 12$, conjecturally universal).

\section*{Acknowledgments}

[Acknowledgments will be added in the final version.]


\end{document}